\documentclass[11pt]{article}

\usepackage{authblk}
\usepackage{amsmath,amssymb,amsthm}
\numberwithin{equation}{section}
\newtheorem{theorem}{Theorem}
\newtheorem{proposition}{Proposition}
\newtheorem{lemma}{Lemma}
\numberwithin{theorem}{section}
\numberwithin{proposition}{section}
\numberwithin{lemma}{section}
\usepackage{mathtools,cancel}
\usepackage{amsfonts}
\usepackage{xcolor}
\usepackage{xparse}
\usepackage{tikz-cd}
\usepackage{graphicx}
\usepackage{float}
\usepackage[margin=0.99in]{geometry}
\usepackage{indentfirst}
\usepackage[skip=10pt plus1pt, indent=20pt]{parskip}
\usepackage{authblk}
\usepackage{hyperref}
\usepackage[none]{hyphenat}
\hypersetup{
    colorlinks,
    citecolor=black,
    filecolor=black,
    linkcolor=black,
    urlcolor=black
}

\usepackage{nomencl}
\makenomenclature

\title{\Large \bfseries{\uppercase {Continuation Criterion For Solutions To The Einstein Equations}}\par}
\author[1,2]{\normalsize Oswaldo Vazquez\thanks{ ovazquez@college.harvard.edu}}
\author[2,3]{\normalsize Puskar Mondal\thanks{ puskar\_mondal@fas.harvard.edu}}
\affil[1]{\itshape \small Department of Physics, Harvard University, 17 Oxford Street, MA 02138, USA}
\affil[2]{\itshape \small Department of Mathematics, Harvard University, 1 Oxford Street, MA 02138, USA}
\affil[3]{\itshape \small Center of Mathematical Sciences and Applications, Harvard University, 20 Garden Street, MA 02138, USA}
\date{}

\begin{document}

\maketitle

\begin{abstract}
    \noindent We prove a continuation condition in the context of 3+1 dimensional vacuum Einstein gravity in Constant Mean extrinsic Curvature (CMC) gauge. More precisely, we obtain quantitative criteria under which the physical spacetime can be extended in the future indefinitely as a solution to the Cauchy problem of the Einstein equations given regular initial data. In particular, we show that a gauge-invariant $H^{2}$ Sobolev norm of the spacetime Riemann curvature remains bounded in the future time direction provided the so-called deformation tensor of the unit timelike vector field normal to the chosen CMC hypersurfaces verifies a spacetime $L^\infty$ bound. To this end, we implement a novel technique to obtain this refined estimate by using Friedlander's parametrix for tensor wave equations on curved spacetime and Moncrief's subsequent improvement \cite{integral,moncrief2019could} \footnote{We note that this continuation criterion was obtained by \cite{klainerman2010breakdown} using a different technique. Our approach has the potential to be used in gravity coupled to other sources and also for addressing large data global existence problems such as the $U(1)$ problem. The current article can be thought of as a warm-up for more difficult scenarios. In fact, one of us recently proved the stability of $3+1$ dimensional Milne spacetime under coupled Einstein-Yang-Mills perturbations where the technique introduced in this article plays an important role- this result is being written up}. We conclude by providing a physical explanation of our result as well as its relation to the issues of determinism and weak cosmic censorship.
\end{abstract}

\medskip


\section{Introduction}
\noindent  Given regular initial data for gravity, it is of mathematical and physical interest to obtain analytic criteria that codify whether the gravitational field will evolve to a unique  singularity-free (naked) global solution to the Einstein equations, which are a system of quasi-linear hyperbolic PDEs while expressed in a suitable gauge (e.g. spacetime harmonic gauge \cite{choquet2008general} or constant mean curvature spatial harmonic gauge \cite{elliptichyperbolic}). The appeal for mathematicians is obvious as there is a plethora of literature studying the breakdown of solutions to non-linear field equations. A few examples where global existence holds in 3+1 Minkowski space are: the nonlinear wave equation $\Box \varphi=\lambda |\varphi|^{p-1}\varphi$ for $1\leq p \leq 4$ \cite{grillakis1990regularity,struwe1988globally, jorgens1961anfangswertproblem, rauch1981u5}, sine-Gordon equation \cite{sine} non-Abelian Yang-Mills-Higgs for a fixed choice of compact gauge group and specific restrictions on the Higgs potential \cite{eardley1982global, eardley1982global2}. Yang-Mills fields on a globally hyperbolic background are also known to exhibit a non-blow-up characteristic \cite{satah,ghanem}. In contrast, some equations that may have finite time blow-ups in 3+1 dimensions include: wave maps (also known as non-linear sigma models) \cite{krieger2008renormalization}, and relativistic perfect fluids \cite{christodoulou2007formation}. Any physically acceptable classical field is desired to be globally well-posed on any fixed globally hyperbolic spacetime.

\noindent From the physics perspective, the breakdown/continuation of solutions to Einstein's equations is essential to (dis)validate the deterministic essence of classical general relativity theory. For if the maximal Cauchy development of regular initial data were to not be regular (in a suitable sense), then the future cannot be fully predicted even with perfect knowledge about the present (therefore a loss of information occurs). A more pathological situation would be if the evolved spacetime contains so-called naked singularities which in principle should be observable by a timelike observer located in the future. These irregularities such as naked singularities and Cauchy horizons are hypothesized to be absent from nature (or exist as purely mathematical objects which are to be unstable against perturbations) as stated in Penrose's cosmic censorship conjecture \cite{penrose1999question}.

\noindent Even though several interesting results in the context of small data global well-posedness problems have been established over the past thirty years \cite{christodoulou1993global,bieri2010extension,zipser2000global, lefloch2016global2, bigorgne2021asymptotic, liu2021new, liu2022global, fajman2021attractors, andersson2020nonlinear, branding2019stable, mondal2022radiation}, the large data problem is far from being solved. While this is an extremely difficult issue to handle if one does not impose a smallness condition on the size of the data, it is a desirable first step to deduce the condition under which the Cauchy problem with large data is indeed globally well-posed in the absence of any symmetry. Another important motivation lies in the field of numerical relativity. If one attempts to solve the initial value problem for Einstein's field equations numerically, then a choice of coordinates must be explicitly made (there are infinitely many possible coordinates due to diffeomorphism invariance). During numerical evolution, it is indeed possible to encounter \textit{fake} singularities which are a result of the choice of coordinates. One such example occurs when working with Gaussian normal coordinates \cite{tipler1977nature}. Coordinate singularities are difficult to handle numerically and they can be wrongfully interpreted as intrinsic spacetime singularities. In order to avoid such issue, one needs sharp criteria that can distinguish between the true and fake singularities.

\noindent  In this paper, we attempt to find such breakdown/continuation condition for the special case of vacuum gravity, namely the equation
\begin{equation}
\label{eq:vacuum}
    \mathbf{Ric}_{\widehat{g}}=0
\end{equation}
 is assumed to be satisfied everywhere in a time orientable globally hyperbolic spacetime $M$ endowed with a  metric tensor $\widehat{g}$ of Lorentzian signature. We will work in Constant Mean extrinsic Curvature Spatial Harmonic (CMCSH) gauge.

\noindent  Previous accounts on continuation criteria for vacuum include the work of Anderson \cite{anderson2001long}, who showed that a breakdown occurs when the $L^{\infty}_{t}L^{\infty}_{\vec{x}}$ norm of the Riemann tensor $\mathbf{Rm}_{\widehat{g}}$ of the spacetime $(M,\widehat{g})$ blows up. An improved breakdown criteria requiring one less degree of differentiability in the metric $\widehat{g}$ was later found by Klainerman and Rodnianski \cite{klainerman2010breakdown}, namely the $L^{\infty}_{t}L^{\infty}_{\vec{x}}$ blow-up of the deformation tensor $^{\mathbf{n}}\pi:=\mathcal{L}_{\mathbf{n}}\widehat{g}$ of the unit timelike vector field $\mathbf{n}$ normal to a Constant Mean Curvature (CMC) foliation of $M$ (here $\mathcal{L}$ is the Lie derivative). This in fact did not require all derivatives of the metric (that would be non-geometrical) but rather certain components describing the extrinsic geometry of the chosen Cauchy hypersurfaces. The work in \cite{klainerman2010breakdown} was extended by Shao \cite{shao2011breakdown} to apply for Einstein-scalar field and Einstein-Maxwell spacetimes. A further improvement due to Wang \cite{wang} required the $L^{1}_{t}L^{\infty}_{\vec{x}}$ norm of $^{\mathbf{n}}\pi$ to be bounded for the continuation of the vacuum CMC foliation. Our treatment differs substantially from that of the previous ones and attempts to shed light on the advantages of considering the Cartan/tetrad/frame bundle formalism that can essentially handle any gauge-covariant tensor wave equations (e.g., Einstein's equations, Yang-Mills equations). 

\noindent The outline of this presentation is as follows. \hyperlink{section.2}{Sec. 2} will lay out the notations and definitions to be used. \hyperlink{section.3}{Sec. 3} states the main theorem and provides examples of the continuation criteria for two different spacetimes (Taub-NUT and closed FLRW universe), the methods of proof are also summarized. In particular, we will need two results due to Chen and LeFloch \cite{chen2008injectivity,lefloch2008injectivity} and the analysis in \cite{KR1,KR2,KR3} regarding lower bounds for the null and chronological injectivity radii of the exponential map at a point $p\in M$ in order to invoke the representation formula for the Riemann curvature deduced by Friedlander \cite{Freidlander} and Moncrief \cite{integral}. \hyperlink{section.5}{Sec. 5} provides such representation formula after a review of the Cartan/tetrad/frame bundle formalism of GR. The full proof of the main theorem is the subject of \hyperlink{section.4}{Sec. 4} and \hyperlink{section.6}{Sec. 6}, we use a bootstrap technique to bound the spacetime $L^\infty$ norm of the Riemann tensor which turns out to be sufficient to conclude well-posedness through usual elliptic arguements in a CMC spatial harmonic or CMC-spatially transported gauge. Concluding remarks are made in \hyperlink{section.7}{Sec. 7}.

\section{Notations and definitions}

\noindent  Let $(M,\widehat{g})$ be a time orientable globally hyperbolic spacetime of dimension 3+1. Since it is globally hyperbolic, $M$ can be foliated by a family of Cauchy hypersurfaces $\{\Sigma_t\}$ as level sets of a time function $t$. The spacetime topology is decomposed as $\Sigma \times \mathbb{R}$. Each level set $\Sigma_{t_1}$ is diffeomorphic to a future one $\Sigma_{t_2}$ thanks to the flow generated by
\begin{equation}
    \begin{gathered}
        \partial_t=N\mathbf{n}+X,
    \end{gathered}
\end{equation}
where $\mathbf{n}$ is a future unit vector field orthogonal to the level sets, $N$ is called the lapse function, and $X$ is called the shift vector field which is tangent to the Cauchy hypersurfaces. Choosing local coordinates $(t,x^i)$ gives the metric in ADM form
\begin{equation}
    \begin{gathered}
        \widehat{g} = -N^2dt \otimes dt + g_{ij}(dx^i + X^idt) \otimes (dx^j + X^jdt),
    \end{gathered}
\end{equation}
here $g_{ij}=g(\partial_i,\partial_j)=\widehat{g}(\partial_i,\partial_j)$ is the induced Riemannian metric on $\Sigma$. Direct calculation shows that 
\begin{equation}
\label{eq:unit_timelike_normal}
    \mathbf{n}=-Ndt
\end{equation}
\begin{equation}
\label{eq:induced_volume}
    \mu_{\widehat{g}}=N\mu_{g}.
\end{equation}
The second fundamental form of the constant time hypersurfaces is defined as $k_{ij}=(\frac{1}{2}\mathcal{L}_{\mathbf{n}}g)_{ij}$. One has a valid Cauchy problem upon choosing an initial level set with suitable regularity conditions (choosing a slicing of the spacetime or equivalently choosing a gauge). In this article we are interested in spacetimes that are foliated by closed Cauchy hypersurfaces of negative Yamabe type (see \cite{schoen1984conformal, moncrief2019could} for detail about negative Yamabe manifolds).

\noindent \textbf{Other objects and spaces relevant to our study}

\noindent $L^{p}(X)$ \quad Lebesgue function space of $p$th-order over the manifold $X$ with density $\mu$. The norm is given by $\lvert\lvert f\rvert\rvert_{L^p}:=(\int_{X}\lvert f\rvert^p\mu)^{1/p}$.

\noindent $L^{\infty}(X)$ \quad Space of measurable functions that are bounded almost everywhere. The supremum norm is $\lvert\lvert f\rvert\rvert_{L^\infty}:=\sup_{x\in X}\lvert f(x)\rvert$. The spacetime $L^{\infty}$ with coordinates $(t,\vec{x})$ is denoted as $L^\infty_t L^\infty_{\vec{x}}$.

\noindent $H^s$ \quad Sobolev space of order $s$ defined on a Cauchy hypersurface.

\noindent $\exp_p:V\subset T_{p}M\rightarrow M$ \quad Exponential map at $p\in M$ with domain $V$ a neighborhood of the origin.

\noindent $\mathrm{Inj}(M,p,E)$\quad Injectivity radius of $\exp_p$ with respect to a Riemannian metric $E$ at $p$, it is defined as the largest positive number $r$ such that $\exp_p$ restricted to the $E$-ball $B_E(0,r)=\{v\in T_pM:E(0,v)\leq r\}$ is a diffeomorphism.

\noindent $\mathrm{NullInj}_{\widehat{g}}(M,p,E)$\quad Null injectivity radius of $\exp_p$ with respect to a Riemannian metric $E$ and Lorentzian metric $\widehat{g}$. It is the largest $r$ such that $\exp_p|_{\mathcal{D}_E(0,r)}$ is a diffeomorphism where $\mathcal{D}_E(0,r)$ is the intersection of $B_E(0,r)$ with the bottom portion of the double null cone $\{v\in T_pM:\widehat{g}(v,v)= 0\}$.

\noindent $D_p$ \quad Image of $\exp_p|_{\mathcal{D}_E(0,r)}$ where $r$ is the null injectivity radius at $p$. Throughout the paper, we call it the ``full past light cone of $p$". Since it extends down to Euclidean distance $r$, $D_p$ will meet a family of spacelike Cauchy level sets $\{\Sigma_t\}_{t\in [t_p-r,t_p]}$ where $t_p$ is the global time coordinate of the point $p$. Declare $B_p(t):=D_p\cap \Sigma_t$ for each $t\in [t_p-r,t_p]$, this is a ball on the level set of $t$ with topological 2-sphere boundary $\sigma_p(t)$. 

\noindent $C_p$ \quad  Mantle of full past light cone of a point $p\in M$, by mantle we mean it does not include the interior of the cone nor the interior of the ball at $t_p-r$.

\noindent $J_p$ \quad Interior of past light cone of a point $p\in M$ plus the ball at $t_p-r$. Note that $D_p:=C_p\cup J_p$, $\sigma_p(t):=C_p\cap\Sigma_{t}$ and $\partial D_p=C_p \cup B_p(t_p-r)$.

\noindent $\chi, \Bar{\chi}$ \quad Null second fundamental forms. Precisely, given two light-like/null future-directed vector fields $L,\Bar{L}$ (corresponding to outgoing and incoming directions, respectively) which are perpendicular to each $\Sigma_t$, define $\chi(v,w):=\widehat{g}(\nabla_v L,w)$ and $\Bar{\chi}(v,w):=\widehat{g}(\nabla_v \Bar{L},w)$ for any two vector fields $v,w\in T\Sigma$. The traces in particular control the evolution of surface areas from spacelike spheres along the incoming and outgoing null directions.

\noindent $\mu_{\widehat{g}}(x):=\sqrt{-\det \widehat{g}(x)}$ \quad Canonical volume form for the Lorentzian manifold $(M,\widehat{g})$.

\noindent $\eta_{ab}$ \quad Minkowski metric.

\noindent $\delta_{ab}$ \quad Kronecker delta.

\section{Main theorem and idea of the proof}
 
\noindent Global hyperbolicity allows us to cast the Einstein equations as a dynamical system with phase space coordinates $(g(t),k(t))$. Assume $\Sigma_{t=0}$ is the initial Cauchy hypersurface and on it we prescribe the data $(g_0,k_0)\in H^{s}\times H^{s-1}$ $(s\geq 4)$ that verifies the constraint equations. One may now begin to study the determinism of the system i.e. we want to understand if the solutions to the evolution equations can be extended to the future without any obstruction or if there are any obstructions then we want to understand their nature. It is not clear how to proceed at this point. We utilize the physical meaning of timelike Killing fields. Existence of a timelike Killing vector field implies that the spatial hypersurface is stationary or the induced geometry does not change along the flow of this vector field. Therefore, if a timelike Killing field exists, then from a physical perspective, the predictability should trivially hold since the data is not changing in time (in a rigorous sense this is tied to Noether's theorem and conservation laws). But a generic spacetime is almost always not stationary. Therefore, we do not have a timelike Killing vector field. However, one can claim that in order for predictability to hold one does not require exact preservation of the initial information but instead \textit{non-drastic change}. In other words, the obstruction to the existence of a timelike Killing field is not infinitely large so that the initial information is not completely deformed within a finite time interval in the future. In order for such a property to hold, a physically plausible guess would be that the gauge-invariant $L^\infty$ norm of the deformation tensor of the unit timelike vector field orthogonal to the Cauchy foliation $^{\mathbf{n}}\pi:=\mathcal{L}_{\mathbf{n}}\widehat{g}$ should remain finite. Indeed, if $^{\mathbf{n}}\pi$ vanishes then $\mathbf{n}$ is Killing and the spacetime is stationary.

\noindent The next question is how do we show that this criteria of finiteness of $^{\mathbf{n}}\pi$ can be used to conclude that the CMCSH vacuum Cauchy problem is globally well-posed in the future? In fact, an established result (e.g., see \cite{andersson2004future}) states that non-blow up of the $H^2(\Sigma)$ norm of Riemann tensor in finite time leads to well-posedness. This is expected from a physical point of view since Riemann curvature is the manifestation of vacuum gravity (recall the geodesic deviation equation) and exhausts all degrees of freedom. The problem at hand has now been reduced to the following: Suppose $T^*$ is the maximal time of existence of a solution, then the solution can actually be extended indefinitely past $T^*$ provided we use the mild assumption of $L^\infty_t L^\infty_{\Vec{x}}$ finiteness of $^{\mathbf{n}}\pi$ to control $||\mathbf{Rm}_{\widehat{g}}||_{H^{2}}$ at $T^*+\epsilon,~\epsilon>0$ in terms of the initial data at $t=0$. Success means the validity of the following

\noindent\textbf{Main Theorem:} \textit{Let $(M=\Sigma\times \mathbb{R},\widehat{g})$ be a globally hyperbolic spacetime and $\Sigma_{t=0}$ be an initial Cauchy hypersurface of negative Yamabe type and on it $(g_{0},k_{0})\in H^{s}\times H^{s-1} ~(s\geq 4)$ is the initial data for the Cauchy problem of the vacuum Einstein evolution equations in Constant Mean extrinsic Curvature Spatial Harmonic (CMCSH) gauge satisfying the constraint equations. This CMCSH Cauchy problem is well posed in $\mathcal{C}\left([0,t^{*}];H^{s}\times H^{s-1}\right)$. In particular, there exists a time $t^{*}>0$ dependent on $||g_{0}||_{H^{s}},||k_{0}||_{H^{s-1}}$ such that the solution map $(g_{0},k_{0})\mapsto (g(t),k(t),N(t),X(t))$ is continuous 
\begin{eqnarray}\nonumber
H^{s}\times H^{s-1}
\to H^{s}\times H^{s-1}\times H^{s+1}\times H^{s+1}.
\end{eqnarray}
Let $T^{*}$ be the maximal time of existence (i.e., $T^{*}\geq t^{*}$) of a solution to the CMCSH Cauchy problem with data $(g_{0},k_{0})$, then either $T^{*}=\infty$ or 
\begin{eqnarray}
\lim_{t\to T^{*}}\sup ||^\mathbf{n}\pi(t)||_{L^{\infty}(\Sigma_{t})}=\infty.
\end{eqnarray}}

\noindent The proof is \textit{not} straightforward. Direct energy-type argument for the Riemann curvature alone fails to yield the desired result since one would require a point-wise bound to close such an argument (which by means of Sobolev embedding can only lead to an existence result for a short time). Therefore, one must require an additional means to estimate the point-wise behavior of the Riemann curvature. To this end, we utilize the integral equation for the Riemann curvature derived by Moncrief \cite{integral}. In order to make sense of such an integral equation on a dynamical spacetime, one needs to have a well-defined geodesically convex neighborhood. This in turn requires a bound on the point-wise norm of the Riemann curvature. Once again, this leads to a circular argument. In order to circumvent this issue, we shall use the bootstrap method (note that this `bootstrap' has nothing to do with the conformal `bootstrap' of field theory) that lies at the heart of hyperbolic partial differential equations \cite{ringstrom2009cauchy}. We begin by assuming a point-wise upper bound of the Riemann curvature. This together with the bound on $^{\mathbf{n}}\pi$ and Theorems 3.3 and 5.2 of \cite{lefloch2008injectivity} allows us to utilize Moncrief's integral equation (\ref{eq:integral}). Analysis of this equation yields a point-wise bound of the Riemann curvature that is better than the bound we assumed in the first place. Upon closure of the bootstrap argument, the newly acquired point-wise bound of Riemann curvature is then used to prove the main theorem via energy estimates.

\noindent We shall now sketch the plan of action in more detail. First and foremost we need to define the gauge-invariant $L^\infty$ norm, it will be with respect to the Riemannian metric below
\begin{align}
\label{eq:euclidean_metric}
    E:=\widehat{g}+2\mathbf{n}\otimes \mathbf{n}.
\end{align}
The norm-squared of a continuous tensor field $\mathbf{T}$ is to be the contraction with itself via $E$
\begin{align}
    |\mathbf{T}|^2_E:=\mathbf{T}^{ab\cdots}~_{cd\cdots}\mathbf{T}^{ef\cdots}~_{gh\cdots}E_{ae}E_{bf}\cdots E^{cg}E^{dh}\cdots
\end{align}
Note that we will sometimes omit the subscript $E$ for cleanliness. We then define the gauge-invariant $L^\infty$ norm over a CMC slice $\Sigma_t\subset M$ as expected
\begin{align}
    ||\mathbf{T}(t)||_{L^\infty(\Sigma_t)}:=\sup_{\Vec{x}\in\Sigma_t} |\mathbf{T}(t,\Vec{x})|_E
\end{align}
Our continuation criteria is $L^\infty$ finiteness of the deformation tensor of $\mathbf{n}$ over the slab $[0,T^*]\times\Sigma$ for any $T^{*}<\infty$
\begin{align}
\label{eq:breakdown_criteria}
    \sup_{t\in[0,T^*]}||^\mathbf{n}\pi(t)||_{L^\infty(\Sigma_t)}< \infty
\end{align}
From here on, we denote by $C(^\mathbf{n}\pi)$ any bounded function that solely depends on the estimate above. 

\noindent With this in hand, one can begin performing the estimates. We require three definitions for total energy over a CMC slice at time $t$, these are to be denoted by $E^{0}(t), E^{1}(t), E^{2}(t)$ and must control the $L^2(\Sigma_t)$ norm-squared of $\mathbf{Rm}_{\widehat{g}}, D\mathbf{Rm}_{\widehat{g}}, D^2\mathbf{Rm}_{\widehat{g}}$ respectively (the spacetime gauge covariant derivative $D$ is to be defined later). Appropriate choices of energies will lead to the following bounds
\begin{align}
\label{eq:schematic_0th}
    &E^{0}(t_2)\leq  e^{C(^\mathbf{n}\pi)|t_2-t_1|}E^{0}(t_1)\leq  \cdots \leq  C(^\mathbf{n}\pi,t_2)E^{0}(0) \\
\label{eq:schematic_1st}
    &E^{1}(t_2)\leq  e^{C(^\mathbf{n}\pi)|t_2-t_1|}\Big(E^{1}(t_1)+\int_{t_1}^{t_2}E^{0}(t)||\mathbf{Rm}_{\widehat{g}}(t)||^2_{L^{\infty}(\Sigma_t)}dt\Big) \\
\label{eq:schematic_2nd}
    &E^{2}(t_2)\leq  e^{C(^\mathbf{n}\pi)|t_2-t_1|}\Big(E^{2}(t_1)+\int_{t_1}^{t_2}E^{1}(t)||\mathbf{Rm}_{\widehat{g}}(t)||^2_{L^{\infty}(\Sigma_t)}dt\Big)
\end{align}
where $t_1\leq t_2$ are times in the interval $[0,T^*]$. Obviously, these inequalities are a consequence of general energy nonconservation as we do not assume a timelike Killing field in the spacetime. Furthermore, the appearance of the $ L^\infty$ norm of the Riemann curvature in the 1st and 2nd order bounds prevents the continuation proof from being trivial, if this term was absent then we can simply repeat the estimates until we reach the initial data. We somehow need to bound this point-wise term by the energies at earlier times in order to apply an iteration argument. To this end, we invoke Moncrief's local integral equation for the Riemann tensor \cite{integral} which holds in the geodesically convex causal domain about a chosen point $p\in M$ (denoted $\mathcal{G}_p$). The existence of this neighborhood relies on lower bounds for the null and chronological injectivity radii of the exponential map about $p$. In the CMC gauge, this is guaranteed by the results of Chen and LeFloch \cite{chen2008injectivity,lefloch2008injectivity} as well as the prior study of Klainerman and Rodnianski \cite{KR1,KR2,KR3} on causal geometry of vacuum spacetime.

\begin{theorem}[Chen and LeFloch]
    \label{thm:chron_inj}
    Let $p\in M$. Suppose the domain of the exponential map $\exp_{p}$ contains an $E-$ball of radius $r$ $B_E(0,r)=\{v\in T_pM:E_p(0,v)\leq r\}$ and the Riemann curvature satisfies
\begin{equation}
    \begin{gathered}
        \sup_{\gamma}\sup_a\;\lvert \mathbf{Rm}_{\widehat{g}}(\gamma(a))\rvert_{E}\leq \frac{1}{r^{2}},
    \end{gathered}
\end{equation}
where supremum is taken over every $\widehat{g}-$geodesic $\gamma$ initiating from a vector lying in $B_{E}(0,r)$, then there exists a uniform constant $C\in (0,1)$ such that the following bound is fulfilled by the injectivity radius
\begin{equation}
    \begin{gathered}
        \frac{\mathrm{Inj}_{\widehat{g}}(M,p,E)}{r}\geq C \frac{\mathrm{Vol}_{\widehat{g}}(\mathcal{B}_{E}(p,Cr))}{r^{4}}
    \end{gathered}
\end{equation}
with $\mathcal{B}_{E}(p,r):=\exp_{p}(B_{E}(0,r))$.
\end{theorem}

\begin{theorem}[LeFloch, Klainerman and Rodnianski]
    \label{thm:null_inj}
    The null injectivity radius of an observer located at $p$ in an Einstein vacuum spacetime is uniformly controlled solely in terms of the lapse function, the second fundamental form of the foliation, finite initial $L^{2}$ data of curvature and lower volume bounds on some initial hypersurface.
\end{theorem} 

\noindent In CMC time gauge, one can use the elliptic equation
\begin{eqnarray}
\Delta_{g}N+|k|^{2}N=\frac{\partial \mathrm{tr}_{g}k}{\partial t},\nonumber
\end{eqnarray}
to obtain a point-wise estimate for the lapse function $N$ in terms of the second fundamental form $k$ (and therefore the deformation tensor of $\mathbf{n}$). The initial data is assumed to be finite, therefore we get a lower bound on the null injectivity radius. Consequently, we are now allowed to \textit{draw} a past light cone emanating from $p$ that exists throughout the range of the null exponential map. Our analysis will require us to work in the normal neighborhoods of several points inside the slab $[0,T^*]\times\Sigma$ and the past cones of such points will need to extend to a uniform length $\delta$. We declare
\begin{align}
\label{eq:delta}
    0<\delta\leq  \inf_{p\in\Sigma\times [0,T^*]} \mathrm{NullInj}_{\widehat{g}}(p,E)
\end{align}
in order to meet our needs (see Fig. \ref{fig:1}). Our bootstrap assumption on the curvature is
\begin{align}
\label{eq:bootstrap}
    \sup_{t\in [0,T^*]}||\mathbf{Rm}_{\widehat{g}}(t)||_{L^\infty(\Sigma_t)}\leq \frac{1}{\delta^2} \quad \mathrm{and}\quad \mathrm{Vol}_{\widehat{g}}(\mathcal{B}_{E}(p,C\delta))\geq \frac{\delta^4}{C} \:\: \forall p\in [0,T^*]\times\Sigma
\end{align}

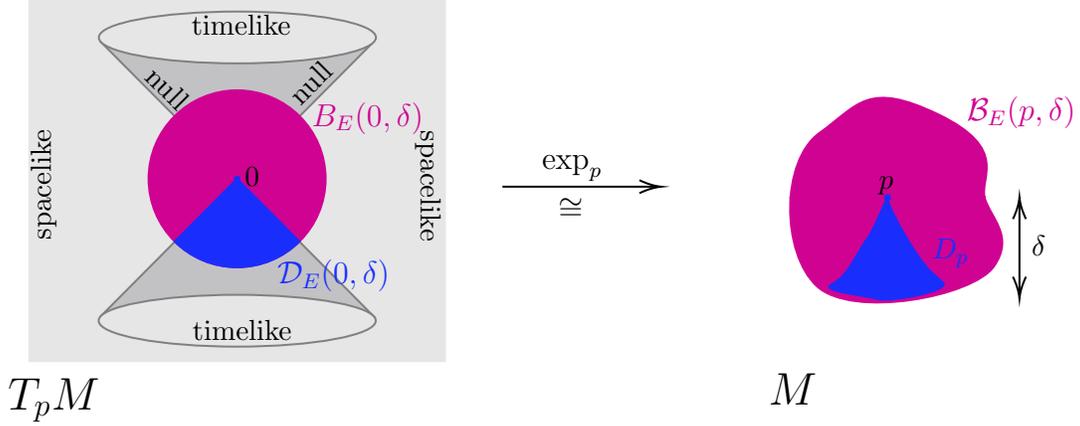
\begin{figure}[t]
    \centering

\tikzset{every picture/.style={line width=0.75pt}} 

\begin{tikzpicture}[x=0.75pt,y=0.75pt,yscale=-1,xscale=1]

\draw  [color={rgb, 255:red, 230; green, 230; blue, 230 }  ,draw opacity=1 ][fill={rgb, 255:red, 230; green, 230; blue, 230 }  ,fill opacity=1 ] (21.42,9.83) -- (230.99,9.83) -- (230.99,193.58) -- (21.42,193.58) -- cycle ;
\draw  [color={rgb, 255:red, 128; green, 128; blue, 128 }  ,draw opacity=1 ][fill={rgb, 255:red, 194; green, 194; blue, 197 }  ,fill opacity=0.67 ] (53.97,29.01) -- (198.45,29.01) -- (53.97,174.4) -- (198.45,174.4) -- cycle ;
\draw  [color={rgb, 255:red, 208; green, 2; blue, 146 }  ,draw opacity=1 ][fill={rgb, 255:red, 208; green, 2; blue, 146 }  ,fill opacity=0.5 ] (81.53,101.56) .. controls (81.53,76.91) and (101.52,56.92) .. (126.17,56.92) .. controls (150.83,56.92) and (170.82,76.91) .. (170.82,101.56) .. controls (170.82,126.22) and (150.83,146.2) .. (126.17,146.2) .. controls (101.52,146.2) and (81.53,126.22) .. (81.53,101.56) -- cycle ;
\draw  [color={rgb, 255:red, 208; green, 2; blue, 146 }  ,draw opacity=1 ][dash pattern={on 1.5pt off 1.5pt}] (81.85,103.48) .. controls (81.85,96.99) and (101.7,91.72) .. (126.17,91.72) .. controls (150.65,91.72) and (170.49,96.99) .. (170.49,103.48) .. controls (170.49,109.97) and (150.65,115.23) .. (126.17,115.23) .. controls (101.7,115.23) and (81.85,109.97) .. (81.85,103.48) -- cycle ;
\draw  [color={rgb, 255:red, 25; green, 46; blue, 253 }  ,draw opacity=1 ][fill={rgb, 255:red, 25; green, 46; blue, 253 }  ,fill opacity=1 ] (125.04,101.71) .. controls (125.04,101.06) and (125.56,100.54) .. (126.21,100.54) .. controls (126.85,100.54) and (127.38,101.06) .. (127.38,101.71) .. controls (127.38,102.35) and (126.85,102.88) .. (126.21,102.88) .. controls (125.56,102.88) and (125.04,102.35) .. (125.04,101.71) -- cycle ;
\draw    (260,105.57) -- (338,105.57) ;
\draw [shift={(340,105.57)}, rotate = 180] [color={rgb, 255:red, 0; green, 0; blue, 0 }  ][line width=0.75]    (10.93,-3.29) .. controls (6.95,-1.4) and (3.31,-0.3) .. (0,0) .. controls (3.31,0.3) and (6.95,1.4) .. (10.93,3.29)   ;
\draw [color={rgb, 255:red, 230; green, 230; blue, 230 }  ,draw opacity=1 ][line width=3.75]    (47.61,29.82) -- (205.04,30.68) ;
\draw  [color={rgb, 255:red, 25; green, 46; blue, 253 }  ,draw opacity=1 ][fill={rgb, 255:red, 25; green, 46; blue, 253 }  ,fill opacity=0.6 ] (126.16,146.2) .. controls (113.9,146.2) and (102.79,141.26) .. (94.72,133.25) -- (126.17,101.56) -- (157.62,133.24) .. controls (149.55,141.25) and (138.43,146.2) .. (126.16,146.2) -- cycle ;
\draw [color={rgb, 255:red, 25; green, 46; blue, 253 }  ,draw opacity=1 ]   (94.72,133.25) .. controls (119.38,125.16) and (133.36,125.16) .. (157.62,133.24) ;
\draw  [color={rgb, 255:red, 128; green, 128; blue, 128 }  ,draw opacity=1 ][fill={rgb, 255:red, 230; green, 230; blue, 230 }  ,fill opacity=1 ] (56.43,30.25) .. controls (56.43,23.03) and (87.72,17.17) .. (126.32,17.17) .. controls (164.93,17.17) and (196.22,23.03) .. (196.22,30.25) .. controls (196.22,37.48) and (164.93,43.33) .. (126.32,43.33) .. controls (87.72,43.33) and (56.43,37.48) .. (56.43,30.25) -- cycle ;
\draw [color={rgb, 255:red, 230; green, 230; blue, 230 }  ,draw opacity=1 ][line width=3.75]    (48.9,173.38) -- (206.33,174.24) ;
\draw  [color={rgb, 255:red, 128; green, 128; blue, 128 }  ,draw opacity=1 ][fill={rgb, 255:red, 230; green, 230; blue, 230 }  ,fill opacity=1 ] (56.33,173.21) .. controls (56.33,165.99) and (87.62,160.13) .. (126.23,160.13) .. controls (164.83,160.13) and (196.12,165.99) .. (196.12,173.21) .. controls (196.12,180.44) and (164.83,186.29) .. (126.23,186.29) .. controls (87.62,186.29) and (56.33,180.44) .. (56.33,173.21) -- cycle ;
\draw  [color={rgb, 255:red, 208; green, 2; blue, 146 }  ,draw opacity=1 ][fill={rgb, 255:red, 208; green, 2; blue, 146 }  ,fill opacity=0.5 ] (507.89,145.78) .. controls (497.22,161.11) and (436.56,172.44) .. (418.56,155.78) .. controls (400.56,139.11) and (400.22,92.56) .. (419.89,78.44) .. controls (439.56,64.33) and (443.89,51.78) .. (479.89,69.11) .. controls (515.89,86.44) and (500.56,101.11) .. (503.22,111.11) .. controls (505.89,121.11) and (518.56,130.44) .. (507.89,145.78) -- cycle ;
\draw  [color={rgb, 255:red, 208; green, 2; blue, 146 }  ,draw opacity=1 ][dash pattern={on 1.5pt off 1.5pt}] (405.22,111.11) .. controls (405.22,104.62) and (427.16,99.36) .. (454.22,99.36) .. controls (481.28,99.36) and (503.22,104.62) .. (503.22,111.11) .. controls (503.22,117.6) and (481.28,122.86) .. (454.22,122.86) .. controls (427.16,122.86) and (405.22,117.6) .. (405.22,111.11) -- cycle ;
\draw  [color={rgb, 255:red, 25; green, 46; blue, 253 }  ,draw opacity=1 ][fill={rgb, 255:red, 25; green, 46; blue, 253 }  ,fill opacity=1 ] (452.91,111.11) .. controls (452.91,110.39) and (453.5,109.8) .. (454.22,109.8) .. controls (454.95,109.8) and (455.53,110.39) .. (455.53,111.11) .. controls (455.53,111.83) and (454.95,112.42) .. (454.22,112.42) .. controls (453.5,112.42) and (452.91,111.83) .. (452.91,111.11) -- cycle ;
\draw  [color={rgb, 255:red, 25; green, 46; blue, 253 }  ,draw opacity=1 ][fill={rgb, 255:red, 25; green, 46; blue, 253 }  ,fill opacity=0.6 ] (426.53,154.57) .. controls (434.37,150.73) and (439.33,137.9) .. (443.53,132.5) .. controls (447.73,127.1) and (454.69,108.67) .. (454.22,111.11) .. controls (453.75,113.56) and (461.61,128.73) .. (468.78,139.23) .. controls (475.94,149.73) and (477.76,149.7) .. (482.13,153.97) .. controls (486.51,158.23) and (454.84,164.19) .. (448.73,162.17) .. controls (442.63,160.15) and (418.7,158.4) .. (426.53,154.57) -- cycle ;
\draw [color={rgb, 255:red, 25; green, 46; blue, 253 }  ,draw opacity=1 ]   (424.97,155.83) .. controls (448.79,149.06) and (454.92,151.94) .. (459.67,152.56) .. controls (464.42,153.19) and (474.92,153.06) .. (482.29,154.73) ;
\draw    (520.75,132.21) -- (520.75,112.96) ;
\draw [shift={(520.75,110.96)}, rotate = 90] [color={rgb, 255:red, 0; green, 0; blue, 0 }  ][line width=0.75]    (10.93,-3.29) .. controls (6.95,-1.4) and (3.31,-0.3) .. (0,0) .. controls (3.31,0.3) and (6.95,1.4) .. (10.93,3.29)   ;
\draw    (520.75,132.21) -- (520.75,159.96) ;
\draw [shift={(520.75,161.96)}, rotate = 270] [color={rgb, 255:red, 0; green, 0; blue, 0 }  ][line width=0.75]    (10.93,-3.29) .. controls (6.95,-1.4) and (3.31,-0.3) .. (0,0) .. controls (3.31,0.3) and (6.95,1.4) .. (10.93,3.29)   ;

\draw (9.27,200.6) node [anchor=north west][inner sep=0.75pt]  [font=\LARGE]  {$T_{p} M$};
\draw (162.49,61.29) node [anchor=north west][inner sep=0.75pt]  [font=\large,color={rgb, 255:red, 208; green, 2; blue, 146 }  ,opacity=1 ]  {$B_{E}( 0,\delta )$};
\draw (128.7,94.21) node [anchor=north west][inner sep=0.75pt]    {$0$};
\draw (278.5,88.4) node [anchor=north west][inner sep=0.75pt]  [font=\large]  {$\mathrm{exp}_{p}$};
\draw (286.75,109.9) node [anchor=north west][inner sep=0.75pt]  [font=\large]  {$\cong $};
\draw (392.75,198.4) node [anchor=north west][inner sep=0.75pt]  [font=\LARGE]  {$M$};
\draw (101.68,17.44) node [anchor=north west][inner sep=0.75pt]   [align=left] {timelike};
\draw (102.22,171.08) node [anchor=north west][inner sep=0.75pt]   [align=left] {timelike};
\draw (230.41,75.67) node [anchor=north west][inner sep=0.75pt]  [rotate=-90] [align=left] {spacelike};
\draw (20.62,133.67) node [anchor=north west][inner sep=0.75pt]  [rotate=-270] [align=left] {spacelike};
\draw (86.17,41.4) node [anchor=north west][inner sep=0.75pt]  [rotate=-45] [align=left] {null};
\draw (148.88,57.73) node [anchor=north west][inner sep=0.75pt]  [rotate=-315] [align=left] {null};
\draw (145.01,140.81) node [anchor=north west][inner sep=0.75pt]  [font=\large,color={rgb, 255:red, 25; green, 46; blue, 253 }  ,opacity=1 ]  {$\mathcal{D}_{E}( 0,\delta )$};
\draw (448.38,98.78) node [anchor=north west][inner sep=0.75pt]    {$p$};
\draw (475.24,130.4) node [anchor=north west][inner sep=0.75pt]  [color={rgb, 255:red, 25; green, 46; blue, 253 }  ,opacity=1 ]  {$D_{p}$};
\draw (492.82,58.62) node [anchor=north west][inner sep=0.75pt]  [font=\large,color={rgb, 255:red, 208; green, 2; blue, 146 }  ,opacity=1 ]  {$\mathcal{B}_{E}( p,\delta )$};
\draw (525.42,128.65) node [anchor=north west][inner sep=0.75pt]    {$\delta $};

\end{tikzpicture}

    \caption{Our assumptions yield a uniform lower bound for the null injectivity radius at any point $p\in\Sigma\times [0,T^*]$ which is controlled by the point-wise bound on the deformation tensor of $\mathbf{n}$ and the regular $L^2$ initial data for the curvature. $\mathcal{D}_E(0,\delta)$ is diffeomorphic to its image under the exponential map, the resulting cone $D_p$ is well-behaved in the sense that light rays emanating from the vertex will not have a common event within the region of existence (Euclidean length $\delta\ll 1$).}
    \label{fig:1}
\end{figure}

\noindent Again, we must eventually justify the above by obtaining a refined point-wise bound. We can now safely invoke Moncrief's integral equation which has a mantle term and a 2-sphere term in the following schematic form 
\begin{align}\nonumber
    \mathbf{Rm}_{\widehat{g}}(x)\sim \int_{C_p}(\,\cdots) + \int_{\sigma_{p}(t_p-\delta)}(\,\cdots), \quad x\in \mathcal{G}_p
\end{align}
\noindent here $p$ is a point in the $[0,T^*]$ slab and $t_p$ is its global time coordinate. The explicit expression is to be given in \hyperlink{section.5}{Sec. 5}. The $L^\infty$ squared norm of the curvature at the $t_p$ Cauchy slice will be shown to obey
\begin{align}
\label{eq:l_infty_schematic}
    ||\mathbf{Rm}_{\widehat{g}}(t_p)||^2_{L^\infty(\Sigma_{t_p})}\leq C(  E^{2}(t_p-\delta)+E^{1}(t_p-\delta)+\delta^{-1}E^{0}(t_p-\delta))
\end{align}
This bound is possible due to two fundamental reasons. Firstly, it will rely on the tetrad/SO(1,3) frame bundle formalism of GR which yields curvature-dependent formulas for the connection and (co-)frame fields in the normal neighborhood of $p$ (and so in $\mathcal{G}_p$) thanks to the defining properties of the bundles at play (exterior covariant derivatives, Cartan structure equations, etc) \cite{integral}\cite{moncrief2019could}. One will then be able to use such formulae in addition to (\ref{eq:schematic_0th})-(\ref{eq:schematic_2nd}) for an application of Grönwall's inequality. Secondly, Moncrief's equation utilizes a so-called \textit{null structure} present in the non-linear gauge-covariant wave equation for the curvature. Long-time existence or finite time blow-up of the quasi-linear hyperbolic Einstein equations is essentially determined by the relative strengths of the non-linearities and the geometric dispersion associated with the wave characteristics (or energy decay caused by a rapid expansion of the spacetime). Oftentimes, the special structure of the non-linearities makes them weak compared to the linear dispersive terms at the level of small data. A large number of studies exist in the literature that deal with this issue of the structure of the non-linearities. Klainerman \cite{klainerman} showed that if the non-linear terms satisfy the so-called \textit{null} condition in $3+1$ dimensions, then the global existence holds for small data limit contrary to a generic non-linearity for which global existence just fails in $3+1$ dimensions. In the current context, the null structure present in some of the $C_p$ non-linear terms is fined-tuned to prevent finite time singularities, namely, there will be no causal focusing of curvature energy due to the absence of Ricatti-type self-interaction (the dispersive effect dominates instead). The precise calculations will be given in \hyperlink{section.6}{Sec. 6}.

\noindent Recall that $p$ has been arbitrary so far, we now declare it to be the point where the supremum of $\mathbf{Rm}_{\widehat{g}}$ is attained over the slab $[t^*-\delta,t^*]\times \Sigma$ (this is possible since the Cauchy slices are closed). Proceed by substituting (\ref{eq:l_infty_schematic}) into the 1st and 2nd-order energy estimates over time intervals of length $\delta$ and iterate until we reach the initial data (see Fig. \ref{fig:2} for a pictorial view of this mechanism), we shall find a pivotal factor of $\delta^{-1}$ which will allow us to close the bootstrap. More precisely, we will obtain
\begin{align}
    \quad E^{2}(t^*)+E^{1}(t^*)+E^{0}(t^*)\lesssim  1 +\delta^{-1}C(^\mathbf{n}\pi,t^*,||\mathbf{Rm}_{\widehat{g}}||_{H^{2}(\Sigma_{t=0})})
\end{align}
where $C(^\mathbf{n}\pi,t^*,||\mathbf{Rm}_{\widehat{g}}||_{H^{2}(\Sigma_{t=0})})$ is a constant dependent \textit{only} on the initial hypersurface $H^2$ data of the curvature, the time $t^*$, and the bound for $^\mathbf{n}\pi$. The next step is to use the above estimate as well as (\ref{eq:l_infty_schematic}) to find the bootstrap refinement. This will then imply that the $L^\infty([0,T^*]\times\Sigma)$ norm of the Riemann tensor is actually bounded by $C(^\mathbf{n}\pi,T^*,||\mathbf{Rm}_{\widehat{g}}||_{H^{2}(\Sigma_{t=0})})$. Therefore, one can finally run the energy estimates (\ref{eq:schematic_1st})-(\ref{eq:schematic_2nd}) over the time interval $[0,t^*]$ and easily reach the data at $t=0$. The conclusive result
\begin{align}\nonumber
    ||\mathbf{Rm}_{\widehat{g}}||_{H^{2}(\Sigma_{t^*})}\leq C(^\mathbf{n}\pi,T^*,||\mathbf{Rm}_{\widehat{g}}||_{H^{2}(\Sigma_{t=0})})< \infty
\end{align}
gives the main theorem through standard elliptic arguments in CMCSH gauge.

\begin{figure}[t]
    \centering

\tikzset{every picture/.style={line width=0.75pt}} 

\begin{tikzpicture}[x=0.75pt,y=0.75pt,yscale=-1,xscale=1]

\draw [color={rgb, 255:red, 0; green, 182; blue, 102 }  ,draw opacity=1 ][line width=0.75]    (171.08,226.58) .. controls (255.85,247.17) and (327.41,216.92) .. (479.08,224.08) ;
\draw [color={rgb, 255:red, 0; green, 182; blue, 102 }  ,draw opacity=1 ][line width=0.75]    (190.33,169.83) .. controls (288.53,180.43) and (390.2,166.17) .. (459,172.17) ;
\draw    (144.67,56.01) .. controls (245.19,19.81) and (456.2,83.74) .. (507.33,56.67) ;
\draw    (150,343) .. controls (245.53,374.73) and (370.73,325.93) .. (499.42,343.06) ;
\draw  [dash pattern={on 5.25pt off 5.25pt}]  (173.73,101.33) .. controls (250.53,72.13) and (369,147.5) .. (477,108.5) ;
\draw [color={rgb, 255:red, 72; green, 19; blue, 254 }  ,draw opacity=1 ][line width=0.75]    (161.33,84.13) .. controls (260.04,45.17) and (382.53,128.4) .. (486.53,85.73) ;
\draw [color={rgb, 255:red, 72; green, 19; blue, 254 }  ,draw opacity=1 ][line width=0.75]    (180.13,136.93) .. controls (270.11,121.32) and (358.07,153.37) .. (467.67,138.17) ;
\draw  [color={rgb, 255:red, 72; green, 19; blue, 254 }  ,draw opacity=1 ][fill={rgb, 255:red, 72; green, 19; blue, 254 }  ,fill opacity=1 ][line width=1.5]  (282.36,123.05) .. controls (282.36,122.04) and (283.18,121.21) .. (284.19,121.21) .. controls (285.2,121.21) and (286.02,122.04) .. (286.02,123.05) .. controls (286.02,124.06) and (285.2,124.88) .. (284.19,124.88) .. controls (283.18,124.88) and (282.36,124.06) .. (282.36,123.05) -- cycle ;
\draw  [color={rgb, 255:red, 0; green, 182; blue, 102 }  ,draw opacity=1 ][fill={rgb, 255:red, 0; green, 182; blue, 102 }  ,fill opacity=1 ][line width=1.5]  (365.07,202.74) .. controls (365.07,201.73) and (365.89,200.91) .. (366.9,200.91) .. controls (367.91,200.91) and (368.73,201.73) .. (368.73,202.74) .. controls (368.73,203.76) and (367.91,204.58) .. (366.9,204.58) .. controls (365.89,204.58) and (365.07,203.76) .. (365.07,202.74) -- cycle ;
\draw  [color={rgb, 255:red, 0; green, 182; blue, 102 }  ,draw opacity=1 ][fill={rgb, 255:red, 0; green, 182; blue, 102 }  ,fill opacity=0.63 ][line width=1.5]  (341.83,252.83) .. controls (346.58,252.36) and (368.83,206.08) .. (366.9,202.74) .. controls (364.97,199.41) and (382.23,253.33) .. (391.48,255.33) .. controls (400.73,257.33) and (373.24,258.74) .. (366.33,258.19) .. controls (359.42,257.65) and (337.08,253.31) .. (341.83,252.83) -- cycle ;
\draw  [color={rgb, 255:red, 72; green, 19; blue, 254 }  ,draw opacity=1 ][fill={rgb, 255:red, 72; green, 19; blue, 254 }  ,fill opacity=0.61 ][line width=1.5]  (259,173) .. controls (265.46,172.38) and (286.12,126.39) .. (284.19,123.05) .. controls (282.26,119.71) and (300.75,170.33) .. (310,172.33) .. controls (319.25,174.33) and (302.87,173.95) .. (282.1,174.1) .. controls (261.33,174.26) and (252.54,173.62) .. (259,173) -- cycle ;
\draw [color={rgb, 255:red, 72; green, 19; blue, 254 }  ,draw opacity=1 ]   (115.33,117) .. controls (151.63,109.74) and (232.36,118.16) .. (277.69,123.74) ;
\draw [shift={(279.05,123.9)}, rotate = 187.05] [color={rgb, 255:red, 72; green, 19; blue, 254 }  ,draw opacity=1 ][line width=0.75]    (10.93,-3.29) .. controls (6.95,-1.4) and (3.31,-0.3) .. (0,0) .. controls (3.31,0.3) and (6.95,1.4) .. (10.93,3.29)   ;
\draw [color={rgb, 255:red, 0; green, 182; blue, 102 }  ,draw opacity=1 ]   (514,190.01) .. controls (434.03,180.99) and (399.52,193.61) .. (373.33,200.14) ;
\draw [shift={(371.73,200.53)}, rotate = 346.37] [color={rgb, 255:red, 0; green, 182; blue, 102 }  ,draw opacity=1 ][line width=0.75]    (10.93,-3.29) .. controls (6.95,-1.4) and (3.31,-0.3) .. (0,0) .. controls (3.31,0.3) and (6.95,1.4) .. (10.93,3.29)   ;
\draw [color={rgb, 255:red, 72; green, 19; blue, 254 }  ,draw opacity=1 ][line width=1.5]    (259,173) .. controls (280.54,167.62) and (282.59,166.34) .. (310,172.33) ;
\draw [color={rgb, 255:red, 0; green, 182; blue, 102 }  ,draw opacity=1 ][line width=1.5]    (341.83,252.83) .. controls (366.6,248.44) and (379.86,252.65) .. (392.49,256.23) ;
\draw  [color={rgb, 255:red, 233; green, 154; blue, 0 }  ,draw opacity=1 ][fill={rgb, 255:red, 233; green, 154; blue, 0 }  ,fill opacity=0.63 ][line width=1.5]  (182.33,349.89) .. controls (187.08,349.41) and (209.5,304.79) .. (207.57,301.46) .. controls (205.63,298.12) and (225.08,350.78) .. (234.33,352.78) .. controls (243.58,354.78) and (212.11,353.89) .. (206.33,353.67) .. controls (200.56,353.44) and (177.58,350.37) .. (182.33,349.89) -- cycle ;
\draw  [color={rgb, 255:red, 233; green, 154; blue, 0 }  ,draw opacity=1 ][fill={rgb, 255:red, 233; green, 154; blue, 0 }  ,fill opacity=1 ][line width=1.5]  (205.73,302.29) .. controls (205.73,301.28) and (206.55,300.46) .. (207.57,300.46) .. controls (208.58,300.46) and (209.4,301.28) .. (209.4,302.29) .. controls (209.4,303.3) and (208.58,304.12) .. (207.57,304.12) .. controls (206.55,304.12) and (205.73,303.3) .. (205.73,302.29) -- cycle ;
\draw [color={rgb, 255:red, 233; green, 154; blue, 0 }  ,draw opacity=1 ][line width=1.5]    (182.33,349.89) .. controls (205.13,346.23) and (214.73,347.89) .. (233.13,353.23) ;

\draw (500.67,335.35) node [anchor=north west][inner sep=0.75pt]  [font=\normalsize]  {$\Sigma _{t=0}$};
\draw (510.86,50.07) node [anchor=north west][inner sep=0.75pt]  [font=\normalsize]  {$\Sigma _{T^{*}}$};
\draw (474.05,100.64) node [anchor=north west][inner sep=0.75pt]  [font=\normalsize]  {$\Sigma _{T^{*} -\delta }$};
\draw (486.33,78.33) node [anchor=north west][inner sep=0.75pt]  [font=\normalsize,color={rgb, 255:red, 72; green, 19; blue, 254 }  ,opacity=1 ]  {$\Sigma _{t^{*}}$};
\draw (460.62,163.93) node [anchor=north west][inner sep=0.75pt]  [font=\normalsize,color={rgb, 255:red, 0; green, 182; blue, 102 }  ,opacity=1 ]  {$\Sigma _{t_{p} -\delta }$};
\draw (277.33,-1.27) node [anchor=north west][inner sep=0.75pt]  [font=\LARGE]  {$\Sigma \times [0,T^{*}]$};
\draw (468.9,129.83) node [anchor=north west][inner sep=0.75pt]  [font=\normalsize,color={rgb, 255:red, 72; green, 19; blue, 254 }  ,opacity=1 ]  {$\Sigma _{t^{*} -\delta }$};
\draw (480.28,215.87) node [anchor=north west][inner sep=0.75pt]  [font=\normalsize,color={rgb, 255:red, 0; green, 182; blue, 102 }  ,opacity=1 ]  {$\Sigma _{t_{p} -2\delta }$};
\draw (360,182.33) node [anchor=north west][inner sep=0.75pt]  [color={rgb, 255:red, 0; green, 182; blue, 102 }  ,opacity=1 ]  {$p'$};
\draw (334.8,205.13) node [anchor=north west][inner sep=0.75pt]  [color={rgb, 255:red, 0; green, 182; blue, 102 }  ,opacity=1 ]  {$D_{p'}$};
\draw (278.1,106.75) node [anchor=north west][inner sep=0.75pt]  [color={rgb, 255:red, 72; green, 19; blue, 254 }  ,opacity=1 ]  {$p$};
\draw (249.01,140.08) node [anchor=north west][inner sep=0.75pt]  [color={rgb, 255:red, 72; green, 19; blue, 254 }  ,opacity=1 ]  {$D_{p}$};
\draw (331.83,257.43) node [anchor=north west][inner sep=0.75pt]  [font=\large,rotate=-90]  {$\cdots $};
\draw (30.73,107.14) node [anchor=north west][inner sep=0.75pt]  [font=\scriptsize,color={rgb, 255:red, 72; green, 19; blue, 254 }  ,opacity=1 ] [align=left] {point where\\$\displaystyle \sup _{x\in \Sigma \times \left[ t^{*} -\delta ,t^{*}\right]} |\mathbf{Rm}_{\hat{g}}( x) |_{E}^{2}$ \\is achieved};
\draw (518.6,169.67) node [anchor=north west][inner sep=0.75pt]  [font=\scriptsize,color={rgb, 255:red, 0; green, 182; blue, 102 }  ,opacity=1 ] [align=left] {point where\\$\displaystyle \sup _{x\in \Sigma \times [ t_{p} -2\delta ,t_{p} -\delta ]} |\mathbf{Rm}_{\hat{g}}( x) |_{E}^{2}$ \\is achieved};
\draw (199.83,279.98) node [anchor=north west][inner sep=0.75pt]  [color={rgb, 255:red, 233; green, 154; blue, 0 }  ,opacity=1 ]  {$p^{\prime\cdots\prime}$};
\draw (159.33,311.48) node [anchor=north west][inner sep=0.75pt]  [color={rgb, 255:red, 233; green, 154; blue, 0 }  ,opacity=1 ]  {$D_{p^{\prime\cdots\prime}}$};
\draw (249.48,56.11) node [anchor=north west][inner sep=0.75pt]  [font=\footnotesize,rotate=-8.5]  {$E^{2}\left( t^{*}\right) +E^{1}\left( t^{*}\right) +E^{0}\left( t^{*}\right)$};
\draw (273.07,333.03) node [anchor=north west][inner sep=0.75pt]  [font=\footnotesize,rotate=-355.13]  {$E^{2}( 0) +E^{1}( 0) +E^{0}( 0)$};

\end{tikzpicture}

    \caption{The iteration mechanism bounds the energies at $t^*$ in terms of the initial data. Since the estimates involve integration over $[t^*-\delta,t^*]$, take out $\mathbf{Rm}_{\widehat{g}}$ as sup norm and let $p$ be the point where it is achieved. Bound the sup norm by means of the light cone mantle estimates (\ref{eq:l_infty_schematic}), they can only go down the length of $D_p$ which is $\delta$ so run more energy bounds starting at $t_p-\delta$. This now includes integration over $[t_p-2\delta,t_p-\delta]$, take out sup norm of Riemann again and let it happen at $p'$. Use light cone estimates with new vertex at $p'$ to bound sup norm by energies at $t_{p'}-\delta$. Repeat a finite number of times until we reach $t=0$.
    \label{fig:2}}
\end{figure}
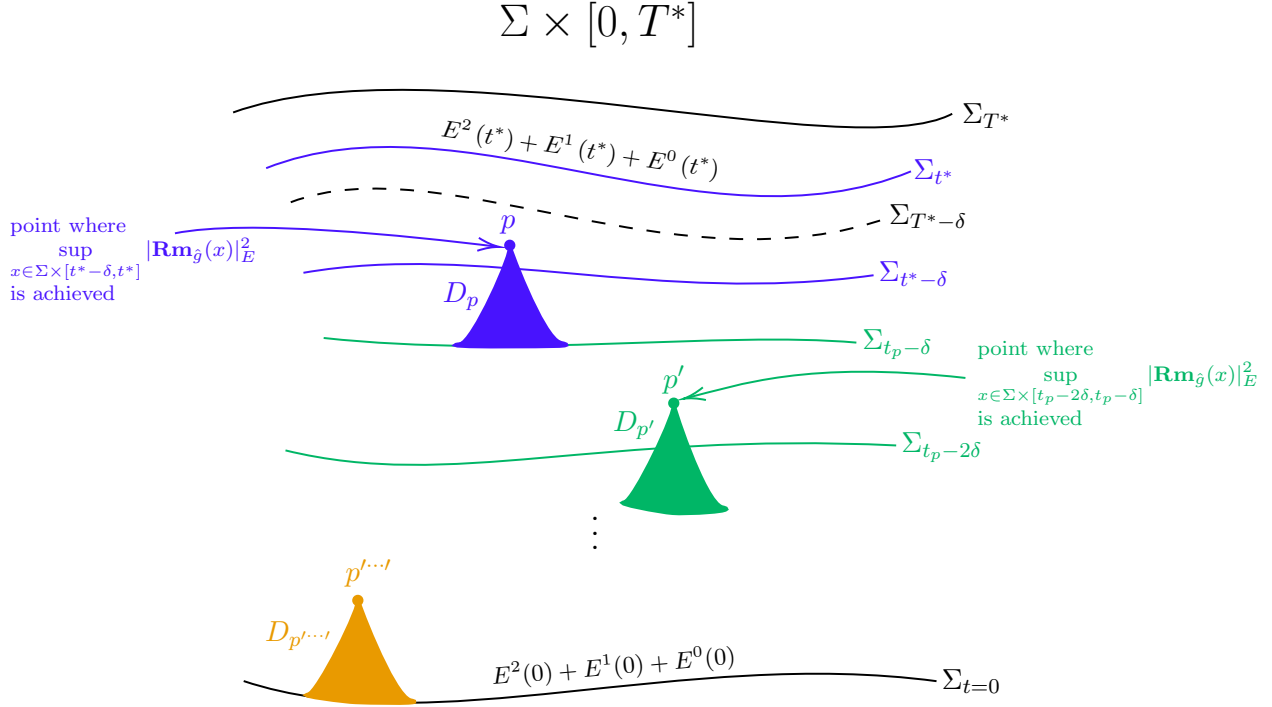

\subsection{Examples of globally hyperbolic spacetimes with compact Cauchy slices that do not satisfy the continuation criteria}

\noindent A first example of a vacuum spacetime for which the gauge-invariant sup norm of $^\mathbf{n}\pi$ is not point-wise bounded is the so-called Taub-NUT solution which has the topology $\mathbb{R}\times S^3$ and metric in Euler coordinates given by
\begin{align}
    \widehat{g}(t,\psi,\theta,\phi)=-U^{-1}(t)dt\otimes dt+&(2l)^2U(t)(d\psi+\cos\theta d\phi)\otimes (d\psi+\cos\theta d\phi) \nonumber \\
    &\qquad+(t^2+l^2)(d\theta\otimes d\theta+\sin^2\theta d\phi\otimes d\phi) \nonumber
\end{align}
where $U(t):=\frac{2mt+l^2-t^2}{t^2+l^2}$ with $m,l$ positive constants. One can immediately conclude that there is no shift $X=0$ and the lapse reads as $N(t)=U(t)^{-1/2}$. Thus, the timelike unit vector field orthogonal to the Cauchy foliation is $\mathbf{n}=U(t)^{1/2}\partial_t$. Direct calculations show that the only non-vanishing coordinate components of the deformation tensor of $\mathbf{n}$ are: $^\mathbf{n}\pi_{\psi\psi},^\mathbf{n}\pi_{\theta\theta},^\mathbf{n}\pi_{\phi\phi},^\mathbf{n}\pi_{\psi\phi}$. The $L^\infty$ norm-squared of $^\mathbf{n}\pi$ with respect to $E=\widehat{g}+2\mathbf{n}\otimes\mathbf{n}$ can be computed to be
\begin{align}
    ||^\mathbf{n}\pi(t)||^2_{L^\infty(S^3_t)}=\frac{U'(t)^2}{U(t)}+\frac{8t^2}{(t^2+l^2)^2}U(t)
\end{align}
which blows-up in finite time at $t_{\pm}=m\pm\sqrt{m^2+l^2}$. Nonetheless, the solution can be extended past these surfaces as shown by Newman, Tamburino, and Unti in 1963.

\noindent Another example is the FLRW metric
\begin{align}
    \widehat{g}=-dt\otimes dt+a^2(t)g_{ij}dx^i\otimes dx^j,
\end{align}
where the scale factor $a(t)$ has the property of going to zero as $t\rightarrow 0$ (the Big Bang). This spacetime has topology $\mathbb{R}\times \Sigma$ ($\Sigma\cong H^{3},\mathbb{E}^{3},$ or $S^{3}$ and their compact quotients) and satisfies Einstein's equations with a perfect fluid source. Let us consider the case $\Sigma\cong S^{3}$. Despite not being a vacuum solution, one would still expect non-blow up of the deformation tensor of $\mathbf{n}$ to give us information regarding the possibility of extending the solution (equally as important, we would also like to examine the strength of the fluid to fully characterize the breakdown condition, this is discussed in our final remark of \hyperlink{section.7}{Sec. 7}). The shift vector field is zero and the lapse equals to unity. Thus, $\mathbf{n}=-dt$ and the only non-vanishing components of the deformation tensor are the spatial ones $^\mathbf{n}\pi_{ij}=2a'(t)a(t)g_{ij}$, this yields
\begin{align}
    ||^\mathbf{n}\pi(t)||^2_{L^\infty(S^3_t)}=12\frac{a'(t)^2}{a(t)^2}
\end{align}
If we denote the pressure and density of the fluid by $P$ and $\rho$ respectively, then we can examine the early time behavior of the scale factor in the matter- and radiation-dominated regimes. The former is characterized by $P=0$ and the condition $a^3\rho=\;$const. in $t$. For $t$ close to 0 the Friedmann equations imply $a(t)\sim 1-\cos(t^{1/3})\sim t^{2/3}$ and $a'(t)\sim t^{-1/3}$, one then sees that the $L^\infty$ norm of $^\mathbf{n}\pi$ goes like $t^{-2}$ and becomes singular at $t=0$. The radiation-dominant regime is described by an equation of state $P=\frac{1}{3}\rho$ and $a^4\rho=\;$const. in $t$. The early time behavior of the scale factor is then $a(t)\sim \sin(t^{1/2})\sim t^{1/2}$ and $a'(t)\sim t^{-1/2}$. Once again, the $L^\infty$ norm of $^\mathbf{n}\pi$ blows up at $t=0$. 

\noindent As alluded to earlier, it is valuable to also examine the strength of the fluid. Concretely, study the gauge-invariant sup norm of the pressure and density in the early time approximation. Do so by revisiting the equations of state. For the matter-dominated universe, $P$ is no issue since it always vanishes but $\rho(t)\sim a(t)^{-3}\sim t^{-2}$ which blows up as one approaches the Big Bang. Likewise, the radiation-dominant case has $P\sim \rho\sim a(t)^{-4}\sim t^{-2}$. The total quantity $||^\mathbf{n}\pi||_{L^\infty}+||(P,\rho)||_{L^\infty}$ 
captures the non-vacuum continuation criteria (at least in the form which we posit at the end of \hyperlink{section.7}{Sec. 7}) and its finiteness is not met at $t=0$, hence no past extension beyond the Big-Bang is possible.

\section{Global energy estimates over the CMC Cauchy slices}

\noindent The ever-present issue of finding a suitable definition of gravitational energy makes its way into our analysis. The reason why there is no clear \textit{a priori} choice is attributed to the very nature of general relativity theory. Concretely, a universal definition of local energy density is not allowed by the equivalence principle (nevertheless, there is a good definition of energy in a quasi-local sense, see \cite{YauWang}). In \hyperlink{section.3}{Sec. 3}, we stated the requirements needed for our purposes, namely three energy definitions which in conjunction will control the $H^2$ squared-norm of the Riemann curvature at a given CMC slice.

\noindent In the absence of sources, the Riemann tensor is equal to the Weyl tensor $\mathbf{W}$ whose role is to capture pure gravity as it is projected out of the Einstein equations. With this in mind, we consider the Bel-Robinson tensor as our 0th order candidate
\begin{equation}
    \begin{split}
        \mathbf{Q}_{\alpha\beta\gamma\delta}&:=\mathbf{W}_{\alpha\mu\gamma\nu}\mathbf{W}_{\beta}~^{\mu}~_{\delta}~^{\nu}+* \mathbf{W}_{\alpha\mu\gamma\nu}~* \mathbf{W}_{\beta}~^{\mu}~_{\delta}~^{\nu}
        \\
        &=\mathbf{Rm}_{\alpha\mu\gamma\nu}\mathbf{Rm}_{\beta}~^{\mu}~_{\delta}~^{\nu}+* \mathbf{Rm}_{\alpha\mu\gamma\nu}~* \mathbf{Rm}_{\beta}~^{\mu}~_{\delta}~^{\nu},
    \end{split}
\end{equation}
where $* \mathbf{W}_{\alpha\beta\gamma\delta}=\frac{1}{2}\epsilon_{\alpha\beta\mu\nu} \mathbf{W}^{\mu\nu}~_{\gamma\delta}$ is the Hodge dual of $\mathbf{W}$. Note that it models a gauge-theoretic stress tensor as it is generally traceless and divergence-free in vacuum
\begin{equation}\nonumber
    \begin{gathered}
        \mathrm{tr}_{\widehat{g}} \mathbf{Q} =0,
        \\
        \mathrm{div}_{\nabla} \: \mathbf{Q}=0.
    \end{gathered}
\end{equation}
The total 0th order energy over $\Sigma_t$ is then defined as
\begin{align}
    E^{0}(t):=\int_{\Sigma_t}\mathbf{Q(n,n,n,n)}\:\mu_{g}
\end{align}
To see that it is positive definite and indeed controls the $L^2(\Sigma_t)$ norm of $\mathbf{Rm}_{\widehat{g}}$ we must set up a frame $(\widehat{L},\widehat{\underline{L}},e_1,e_2)$. $\widehat{L}$ and $\widehat{\underline{L}}$ are null future-directed and determined by the Eikonal equations in correspondence to a double null foliation of the spacetime ($\widehat{L}$ is the incoming direction and $\widehat{\underline{L}}$ is the outgoing, e.g. see \cite{christodoulou2012formation}). Moreover, $e_1$ and $e_2$ are tangent to the $2-$spheres that foliate the null cones. Performing a conformal transformation $L=a^{2}\widehat{L},~\bar{L}=a^{-2}\widehat{\underline{L}}$, $a:M\to\mathbb{R}$, demand the following to be satisfied 
\begin{equation}
    \begin{gathered}
        \widehat{g}(L,e_A)=\widehat{g}(\bar{L},e_A)=0, \quad 
        \widehat{g}(L,\bar{L})=-2,
        \quad
        \widehat{g}(e_A,e_B)=\delta_{AB}
    \end{gathered}
\end{equation}
for $A,B$ running from 1 to 2. The metric in this null frame basis can then be written as
\begin{equation}
\label{eq:null_metric}
    \widehat{g}=-\frac{1}{2}(L\otimes \bar{L}+\bar{L}\otimes L)+e_{1}\otimes e_{1}+e_{2}\otimes e_{2},
\end{equation}
The timelike unit vector field $\mathbf{n}$ orthogonal to the spacelike level sets has an expression in terms of $L$ and $\bar{L}$
\begin{equation}
    \begin{gathered}
        \mathbf{n}\approx \frac{1}{2}(\bar{L}+L),
    \end{gathered}
\end{equation}
where $\approx$ indicates equality modulo a positive function that is uniformly bounded by means of (\ref{eq:breakdown_criteria}). 
Let us explicitly write down the different null components of the Riemann curvature 
\begin{eqnarray}
\alpha_{AB}:=\mathbf{Rm}_{\widehat{g}}(e_{A},L,e_{B},L),~\bar{\alpha}_{AB}:=\mathbf{Rm}_{\widehat{g}}(e_{A},\overline{L},e_{B},\overline{L}),\nonumber\\
2\beta_{A}:=\mathbf{Rm}_{\widehat{g}}(L,\overline{L},L,e_{A}),~
2\bar{\beta}_{A}:=\mathbf{Rm}_{\widehat{g}}(\overline{L},L,\overline{L},e_{A}),\\
\rho:=\frac{1}{4}\mathbf{Rm}_{\widehat{g}}(L,\overline{L},L,\overline{L}),
~e:=\frac{1}{4}{*}\mathbf{Rm}_{\widehat{g}}(L,\overline{L},L,\overline{L}).\nonumber
\end{eqnarray}
We then obtain the following for the energy density 
\begin{eqnarray}
C^{-1}(|\alpha|^{2}+|\underline{\alpha}|^{2}+|\beta|^{2}+|\underline{\beta}|^{2}+|\rho|^{2}+|e|^{2})\leq \mathbf{Q(n,n,n,n)}\leq C (|\alpha|^{2}+|\underline{\alpha}|^{2}+|\beta|^{2}+|\underline{\beta}|^{2}+|\rho|^{2}+|e|^{2})
\end{eqnarray}
where $C$ is a uniform positive constant dependent only on the point-wise bound for $^\mathbf{n}\pi$. Now that we have established control of $||\mathbf{Rm}_{\widehat{g}}(t)||^2_{L^{2}(\Sigma_t)}$ via $E^{0}(t)$, we must do the same for the 1st and 2nd derivatives. It is at this point that we resort to the Cartan/tetrad/SO(1,3)-frame bundle formalism by considering a connection 1-form $\omega$ relative to the Levi-Civita connection $\nabla$ of $(M,\widehat{g})$, the defining local relation is
\begin{align}
    \omega^a~_{b\mu}=\langle\Theta^a,\nabla_\mu h_b\rangle,
\end{align}
where $\{h_a=h^\mu_a\partial_\mu\}_{a,\mu=0}^3$ and $\{\Theta^a=\Theta^a_\mu dx^\mu\}_{a,\mu=0}^3$ are choices of SO(1,3)-orthonormal frame and co-frame fields, respectively\footnote{The $a,b$ indices (known as gauge indices in the physics literature) reflect the fact that $\omega$ is an endomorphism-valued one-form acting on $\mathbb{R}^{4}\cong$ the fibers of the associated vector bundle to the SO(1,3)-frame bundle.}. The connection $\omega$ induces an exterior covariant derivative $d_\omega$ whose local action on the base is to be denoted by $D_\mu$. It acts on the Riemann tensor as
\begin{align}\nonumber
    D_{\alpha}\mathbf{Rm}^{a}~_{b\mu\nu}=\nabla_{\alpha}\mathbf{Rm}^{a}~_{b\mu\nu}+\omega^{a}~_{c\alpha}\mathbf{Rm}^{c}~_{b\mu\nu}-\omega^{c}_{b\alpha}\mathbf{Rm}^{a}~_{c\mu\nu}.
\end{align}
We demand higher order energies to specifically control the $||D^{2}\mathbf{Rm}_{\widehat{g}}(t)||^{2}_{L^{2}(\Sigma_{t})}+||D\mathbf{Rm}_{\widehat{g}}(t)||^{2}_{L^{2}(\Sigma_{t})}$ norm. This will be possible by defining the following ad-hoc stress tensors which emulate a massless scalar field theory
\begin{eqnarray}
\label{eq:stress1}
\mathfrak{T}^{1}_{\mu\nu}:=D_{\mu}\mathbf{Rm}_{\widehat{g}}\cdot D_{\nu}\mathbf{Rm}_{\widehat{g}}-\frac{1}{2}\widehat{g}_{\mu\nu}D_{\alpha}\mathbf{Rm}_{\widehat{g}}\cdot D^{\alpha}\mathbf{Rm}_{\widehat{g}},\\
\label{eq:stress2}
\mathfrak{T}^{2}_{\mu\nu}:=D_{\mu}D\mathbf{Rm}_{\widehat{g}}\cdot D_{\nu}D\mathbf{Rm}_{\widehat{g}}-\frac{1}{2}\widehat{g}_{\mu\nu}D_{\alpha}D\mathbf{Rm}_{\widehat{g}}\cdot D^{\alpha}D\mathbf{Rm}_{\widehat{g}},
\end{eqnarray}
where $\cdot$ denotes the inner product with respect to the Riemannian metric $E$ (\ref{eq:euclidean_metric}).
The energies over $\Sigma_t$ are given by 
\begin{align}
    E^{1,2}(t):=\int_{\Sigma_t}\mathfrak{T}^{1,2}(\mathbf{n,n})\:\mu_{g}
\end{align}
Direct calculation of the densities confirms that the above is positive definite and matches our demands, in particular, we find
\begin{align}
    &\mathfrak{T}^{1}(\mathbf{n,n})  \;\; \mathrm{controls}\;\; |D_{L}\mathbf{Rm}_{\widehat{g}}|^{2}+|D_{\bar{L}}\mathbf{Rm}_{\widehat{g}}|^{2}+\sum_{A=1,2}|D_{A}\mathbf{Rm}_{\widehat{g}}|^{2}
    \\
    &\mathfrak{T}^{2}(\mathbf{n,n}) \;\; \mathrm{controls}\;\; |D_{L}D\mathbf{Rm}_{\widehat{g}}|^{2}+|D_{\bar{L}}D\mathbf{Rm}_{\widehat{g}}|^{2}+\sum_{A=1,2}|D_{A}D\mathbf{Rm}_{\widehat{g}}|^{2}
\end{align}
Computation of the $\nabla$-divergences produces terms with the deformation tensor $^{\mathbf{n}}\pi$ which are unable to be canceled out due to the Riemannian metric $E$ not being compatible with $\nabla$. One gets the following schematic expressions
\begin{align}
\label{eq:div_stress1}
    \mathrm{div}_{\nabla}~\mathfrak{T}^{1}\sim D\mathbf{Rm}_{\widehat{g}}\cdot D\mathbf{Rm}_{\widehat{g}}\cdot \,^\mathbf{n}\pi+D\mathbf{Rm}_{\widehat{g}}\cdot \mathbf{Rm}_{\widehat{g}}\cdot \mathbf{Rm}_{\widehat{g}},\\
\label{eq:div_stress2}
    \mathrm{div}_{\nabla}~\mathfrak{T}^{2}\sim D^{2}\mathbf{Rm}_{\widehat{g}}\cdot D^{2}\mathbf{Rm}_{\widehat{g}}\cdot \,^\mathbf{n}\pi + D^{2}\mathbf{Rm}_{\widehat{g}}\cdot D\mathbf{Rm}_{\widehat{g}}\cdot \mathbf{Rm}_{\widehat{g}}
\end{align}
by means of the gauge wave equation for the curvature
\begin{equation}\nonumber
    D^2\mathbf{Rm}^{{a}}~_{{b}\mu\nu}=2\mathbf{Rm}^{{a}}~_{{c}\mu\beta}\mathbf{Rm}^{{c}}~_{{b}\nu}~^{\beta}-2\mathbf{Rm}^{{a}}~_{{c}\nu\beta}\mathbf{Rm}^{{c}}~_{{b}\mu}~^{\beta}-\mathbf{Rm}^{\gamma}~_{\beta\mu\nu}\mathbf{Rm}^{{a}}~_{{b}\gamma}~^{\beta}
\end{equation}
We re-derive the above in the next section. It is precisely these non-vanishing divergences that will account for the $L^\infty$ of the Riemann term found in the energy estimates from \hyperlink{section.3}{Sec. 3}. We now begin to deduce such estimates.

\begin{proposition}
    \label{prop:general_global_energy_estimates}
    Let $t_1\leq t_2$ be times in $[0,T^*]$ and $\mathfrak{J}$ be either $\mathbf{Q}(\mathbf{n},\mathbf{n},\mathbf{n},\cdot)$ or $ \mathfrak{T}^{1,2}(\mathbf{n},\cdot)$, then the energies at $t_1$ and $t_2$ are related as follows
    \begin{align}
        E(t_2)\leq E(t_1) +C(^\mathbf{n}\pi)\int_{t_1}^{t_2}\int_{\Sigma_t}|\mathrm{div}_{\nabla}\mathfrak{J}|\:\mu_{g} dt
    \end{align}
\end{proposition}
\begin{proof}
    Pick any $t\in[t_1,t_2]$, Stokes' theorem and the assumption that each $\Sigma_t$ is closed implies
    \begin{align}
        \int_{\Sigma_t}\mathrm{div}_{\nabla}\mathfrak{J}\,N\:\mu_{g}dx^1\wedge dx^2\wedge dx^3&=\int_{\Sigma_t}\frac{1}{\mu_{\widehat{g}}}\partial_{\alpha}\big(\mu_{\widehat{g}}\,\mathfrak{J}^{\alpha}\big)N\:\mu_{g}dx^1\wedge dx^2\wedge dx^3 \nonumber \\
        &=\frac{d}{dt}\int_{\Sigma_t}\mathfrak{J}^{t}N\:\mu_{g}dx^1\wedge dx^2\wedge dx^3 \nonumber
    \end{align}
    Notice $\mathfrak{J}^{t}N=-\mathfrak{J}(\mathbf{n})$, therefore we arrive at the result by integrating both sides over $[t_1,t_2]$ and taking $N$ out of the divergence term as a sup norm, which is controlled by $^\mathbf{n}\pi$ in CMC gauge.
\end{proof}

\begin{proposition}
\label{prop:energy_estimates}
    Suppose $t, t_1, t_2$ with $t_1\leq t_2$ are times in the interval $[0,T^*]$, the 0th, 1st, and 2nd order energies verify the following
    \begin{align}
        \label{eq:0th_energy}
        &E^{0}(t)\leq  C(^{\mathbf{n}}\pi,t)E^{0}(0) \\
        \label{eq:1st_energy}
        &E^{1}(t_2)\leq  e^{C(^\mathbf{n}\pi)|t_2-t_1|}\Big(E^{1}(t_1)+C(^{\mathbf{n}}\pi,t_2,E^{0}(0))\int_{t_1}^{t_2}||\mathbf{Rm}_{\widehat{g}}(t)||^2_{L^{\infty}(\Sigma_t)}dt\Big) \\
        \label{eq:2nd_energy}
        &E^{2}(t_2)\leq  e^{C(^\mathbf{n}\pi)|t_2-t_1|}\Big(E^{2}(t_1)+C'(^{\mathbf{n}}\pi)\int_{t_1}^{t_2}E^{1}(t)||\mathbf{Rm}_{\widehat{g}}(t)||^2_{L^{\infty}(\Sigma_t)}dt\Big)
    \end{align}
    In particular, the $L^2$ norm of the curvature over any time slice is bounded by the initial Cauchy data while the higher order norms require point-wise control of $\mathbf{Rm}_{\widehat{g}}$.
\end{proposition}
\begin{proof}
    Begin with the 0th-order estimates, these are quite simple since the Bel-Robinson tensor is divergence-free. Direct computation shows $\mathrm{div}_{\nabla}\mathbf{Q}(\mathbf{n,n,n,}\cdot)=\frac{3}{2}\langle \mathbf{Q}(\mathbf{n,n,}\cdot,\cdot),\,^\mathbf{n}\pi \rangle$, thus application of the last Proposition \ref{prop:general_global_energy_estimates} leads to 
    \begin{align}
        E^{0}(t_2)\leq  E^{0}(t_1)+C(^{\mathbf{n}}\pi)\int_{t_1}^{t_2}E^{0}(t)\:dt
    \end{align}
    Grönwall's inequality implies 
    \begin{align}
        E^{0}(t_2)\leq  e^{C(^{\mathbf{n}}\pi)|t_2-t_1|} E^{0}(t_1)
    \end{align}
    In particular, for $t_2=t$ we directly estimate $L^{2}$ energy at time $t$ in terms of the initial energy $E^{0}(0)$ up to a constant that solely depends on $t$ as well as the estimate for $^{\mathbf{n}}\pi$.
    \begin{align}
        E^{0}(t)\leq  e^{C(^{\mathbf{n}}\pi)t} E^{0}(0)
    \end{align}

    \noindent Proceed with 1st and 2nd orders. Here we find a bit more difficulty as the non-zero divergences of the stress tensors must be controlled. The divergence of the current density is evaluated to be $\mathrm{div}_{\nabla}\mathfrak{J}^{1,2}=\frac{1}{2}\langle \mathfrak{T}^{1,2},\, ^{\mathbf{n}}\pi \rangle+\langle \mathrm{div}_{\nabla}\mathfrak{T}^{1,2}, \mathbf{n} \rangle$. Invoking the expressions (\ref{eq:div_stress1})-(\ref{eq:div_stress2}) yields the following inequalities for the energies
    \begin{align}
        E^{1,2}(t_2)&\leq  E^{1,2}(t_1)+C(^{\mathbf{n}}\pi)\int_{t_1}^{t_2}E^{1,2}(t)\:dt +C'(^{\mathbf{n}}\pi)\int_{t_1}^{t_2}\int_{\Sigma_{t}}|D^{1,2}\mathbf{Rm}_{\widehat{g}}(t,\Vec{x})|^{2}\cdot|^{\mathbf{n}}\pi(t,\Vec{x})|\:\mu_{g}dt \nonumber\\
        &\qquad\qquad+C'(^{\mathbf{n}}\pi)\int_{t_1}^{t_2}\int_{\Sigma_{t}}|D^{1,2}\mathbf{Rm}_{\widehat{g}}(t,\Vec{x})|\cdot|D^{0,1}\mathbf{Rm}_{\widehat{g}}(t,\Vec{x})|\cdot|\mathbf{Rm}_{\widehat{g}}(t,\Vec{x})|\:\mu_{g}dt \\
        &\leq  E^{1,2}(t_1)+C(^{\mathbf{n}}\pi)\int_{t_1}^{t_2}E^{1,2}(t)\:dt+C'(^{\mathbf{n}}\pi)\int_{t_1}^{t_2}E^{0,1}(t)||\mathbf{Rm}_{\widehat{g}}(t)||^{2}_{L^{\infty}(\Sigma_t)}\:dt\nonumber
    \end{align}
    Application of Grönwall's inequality yields the result
    \begin{align}
        E^{1,2}(t_2)\leq  e^{C(^{\mathbf{n}}\pi)|t_2-t_1|}\Big(E^{1,2}(t_1)+C'(^{\mathbf{n}}\pi)\int_{t_1}^{t_2}E^{0,1}(t)||\mathbf{Rm}_{\widehat{g}}(t)||^{2}_{L^{\infty}(\Sigma_t)}\:dt\Big)
    \end{align}
    For the 1st order estimate, the $E^{0}$ present in the integral can be taken out since it is uniformly bounded by $E^{0}(0)$, $t_2$, and $^{\mathbf{n}}\pi$.
\end{proof}

\section{Local integral equations for the connection, co-frame fields, and curvature}

\noindent We now present formulae that will serve a crucial role in estimating the $||\mathbf{Rm}_{\widehat{g}}(t)||^{2}_{L^{\infty}(\Sigma_t)}$ term in the 1st and 2nd order energy estimates. 

\noindent Recall that the connection 1-form $\omega$ depends on the choice of SO(1,3)-orthonormal frame and co-frame. If we let $h_a$ and $\Theta^a$ be initially defined at $p\in M$ and parallel propagate them along the radial geodesics in the normal neighborhood of $p$ (denoted $\mathcal{N}_p$), then this process does not destroy the duality $\langle h_a,\Theta^b\rangle=\delta_a^b$ and orthonormality $\widehat{g}(h_a,h_b)=\eta_{ab}$. The metric thus has the following form in all of $\mathcal{N}_p$
\begin{equation}
\label{eq:metric_SO_frame}
    \widehat{g}=\eta_{ab}\Theta^a\otimes\Theta^b
\end{equation}
Denoting the normal coordinates of a point $x\in \mathcal{N}_p$ by $x^\mu$, the unique geodesic connecting $p$ to $x$ is radial and of the form $\gamma_x(\lambda)^\mu=\lambda x^\mu$ with $\lambda\in [0,1]$. The parallel transport equations for the frame and co-frame are then $\nabla_{\dot{\gamma}_x}h_{a}=0$ and $\nabla_{\dot{\gamma}_x}\Theta^{a}=0$. In particular, we see
\begin{align}
\label{eq:cronstrom}
    0=\langle \Theta^a , \nabla_{\dot{\gamma}_x}h_{b} \rangle=x^\mu\langle \Theta^a , \nabla_{\mu}h_{b} \rangle = x^\mu \omega^{a}~_{b\mu}=\langle x,\omega^{a}~_{b} \rangle.
\end{align}
which emulates the Cronström gauge condition in Yang-Mills theory. Contracting the coordinate of $x$ with the curvature
\begin{align}
    x^\mu \mathbf{Rm}_{\mu\nu}=x^\mu(d\omega +\omega\wedge\omega)_{\mu\nu}
\end{align}
and using (\ref{eq:cronstrom}) along the radial geodesic results in
\begin{align}
    \lambda x^\nu \mathbf{Rm}_{\mu\nu} (\lambda x)=-\frac{d}{d\lambda}\Big(\lambda \omega_\mu (\lambda x)\Big)
\end{align}
Integration over the entire curve parameter yields
\begin{equation}
\label{eq:connection}
    \begin{gathered}
        \omega^{a}~_{b\mu}(x)=-\int_0^1\lambda x^\nu \mathbf{Rm}^{a}~_{b\mu\nu}(\lambda x)\:d\lambda, \quad \forall x\in \mathcal{N}_p
    \end{gathered}
\end{equation}
As for the co-frame, note that because the tangent vectors of the radial geodesics are constant along the curve then their pairing with the parallel propagated co-frame must also be independent of $\lambda$
\begin{align}
\label{eq:co-frame_gauge}
    \langle \Theta^{a}, \dot{\gamma}_x \rangle_{\lambda=0} = \langle \Theta^{a}, \dot{\gamma}_x \rangle_{\lambda=1}
\end{align}
Contract the coordinate of $x$ (equivalently the coordinate function of the geodesic tangent vector) with the torsion-free equation
\begin{align}
\label{eq:torsion_free}
    0=d_\omega \Theta^a=d\Theta^a+\omega\wedge\Theta^a
\end{align}
and make use of (\ref{eq:cronstrom}) and (\ref{eq:co-frame_gauge}) to obtain the following along the radial geodesic
\begin{align}
    0=\frac{d}{d\lambda}\Big(\lambda\Theta^{a}_{\mu}(\lambda x)-\lambda\Theta^{a}_{\mu}(p)\Big)-\omega^{a}~_{b\mu}(\lambda x)\lambda x^\nu\Theta^{b}_{\nu}(p)
\end{align}
Integrate over the curve parameter to conclude
\begin{equation}
\label{eq:co-frame}
    \begin{gathered}
        \Theta^{a}_{\mu}(x)=\Theta^{a}_{\mu}(p)+\int_0^1 \omega^{a}~_{b\mu}(\lambda x)\lambda x^\nu \Theta^{b}_{\nu}(p)\:d\lambda .
    \end{gathered}
\end{equation}
We now deduce a wave equation obeyed by the Riemann tensor. Begin with the second Bianchi identity
\begin{equation}\nonumber
    \begin{gathered}
        d_\omega \mathbf{Rm}=0 \iff D_{\alpha}\mathbf{Rm}^{a}~_{b\mu\nu}+D_{\mu}\mathbf{Rm}^{a}~_{b\nu\alpha}+D_{\nu}\mathbf{Rm}^{a}~_{b\alpha\mu}=0
    \end{gathered}
\end{equation}
Imposition of the vacuum Einstein equations yields the exact Yang-Mills type equation for the curvature  
\begin{eqnarray}
D^{\nu}\mathbf{Rm}^{a}~_{b\mu\nu}=0.
\end{eqnarray}
Using the following commutation relation
\begin{multline*}
    [D_{\alpha},D_{\beta}]\mathbf{Rm}^{a}~_{b\mu\nu}=
    \\
    \mathbf{Rm}^{a}~_{c\alpha\beta}\mathbf{Rm}^{c}~_{b\mu\nu}-\mathbf{Rm}^{c}~_{b\alpha\beta}\mathbf{Rm}^{a}~_{c\mu\nu}-\mathbf{Rm}^{\gamma}~_{\mu\alpha\beta}\mathbf{Rm}^{a}~_{b\gamma\nu}-\mathbf{Rm}^{\gamma}~_{\nu\alpha\beta}\mathbf{Rm}^{a}~_{b\mu\gamma}
\end{multline*}
as well as the vacuum equation and the first Bianchi identity results in a gauge wave equation (a similar wave equation is satisfied by the Yang-Mills curvature with a different gauge group)
\begin{equation}
\label{eq:gaugecovariant}
    D^{\alpha}D_{\alpha}\mathbf{Rm}^{{a}}~_{{b}\mu\nu}=2\mathbf{Rm}^{{a}}~_{{c}\mu\beta}\mathbf{Rm}^{{c}}~_{{b}\nu}~^{\beta}-2\mathbf{Rm}^{{a}}~_{{c}\nu\beta}\mathbf{Rm}^{{c}}~_{{b}\mu}~^{\beta}-\mathbf{Rm}^{\gamma}~_{\beta\mu\nu}\mathbf{Rm}^{{a}}~_{{b}\gamma}~^{\beta},
\end{equation}
Restoring $\nabla$ dependence yields the following
\begin{align}
    \Box \mathbf{Rm}_{\mu\nu} &=[\nabla^{\alpha}\omega_{\alpha},\mathbf{Rm}_{\mu\nu}]+2[\omega_{\alpha},\nabla^{\alpha}\mathbf{Rm}_{\mu\nu}]+[\omega^{\alpha},[\omega_{\alpha},\mathbf{Rm}_{\mu\nu}]] \nonumber \\
    &\quad+2[\mathbf{Rm}_{\mu\beta},\mathbf{Rm}_{\nu}~^{\beta}]-\mathbf{Rm}^{\gamma}~_{\beta\mu\nu}\mathbf{Rm}_{\gamma}~^{\beta}
\end{align}
where $\Box:=\nabla^{\alpha}\nabla_{\alpha}$. Equipped with the lower bounds for the injectivity radii about $p\in \Sigma\times[0,T^*]$ discussed in \hyperlink{section.3}{Sec. 3}, \cite{integral, moncrief2019could} used Friedlander's theory of wave operators \cite{Freidlander} to get the following integral equation for the frame components of the curvature $\mathbf{Rm}^{a}~_{bmn}=\mathbf{Rm}^{a}~_{b\mu\nu}h^\mu_m h^\nu_n$ evaluated at any point $x$ in the geodesically convex neighborhood $\mathcal{G}_p$ about $p$. \footnote{Note that this equation is substantially different from the ones used by \cite{klainerman2010breakdown, wang}.}
\begin{align}\nonumber
\label{eq:integral}
    \mathbf{Rm}^{a}~_{bmn}(x)&=\frac{1}{2\pi}\int_{C_{p}\,\subset\, \mathcal{G}_{p}}\mu_{\Gamma}(x^{'})\Big\{\Big[-\omega^{p}~_{m\sigma}(x^{'})D^{\sigma}(k(x,x^{'})\mathbf{Rm}^{a}~_{bpn}(x^{'}))
    &&\\\nonumber
    &\qquad-\omega^{p}~_{n\sigma}(x^{'})D^{\sigma}(k(x,x^{'})\mathbf{Rm}^{a}~_{bmp}(x^{'}))\Big.\Big.\Big.\Big.-\omega^{c}~_{b\sigma}(x^{'})D^{\sigma}(k(x,x^{'})\mathbf{Rm}^{a}~_{cmn}(x^{'}))
    &&\\\nonumber
    &\qquad+\omega^{a}~_{c\sigma}(x^{'})D^{\sigma}(k(x,x^{'})\mathbf{Rm}^{c}~_{bmn}(x^{'}))\Big]\Big.\Big.+k(x,x^{'})\Big[-2\mathbf{Rm}^{a}~_{cmp}(x^{'})\mathbf{Rm}^{c}~_{bn}~^{p}(x^{'})
    &&\\\nonumber
    &\qquad+2\mathbf{Rm}^{a}~_{cnp}(x^{'})\mathbf{Rm}^{c}~_{bm}~^{p}(x^{'})+\mathbf{Rm}^{a}~_{bpq}(x^{'})\mathbf{Rm}_{mn}~^{pq}(x^{'})\Big]\Big.
    &&\\  
    &\qquad\Big.+\mathbf{Rm}^{a}~_{bmn}(x^{'})\Box k(x,x^{'})+2\nabla^{\sigma}k(x,x^{'})\Big[\omega^{p}~_{m\sigma}(x^{'})\mathbf{Rm}^{a}~_{bpn}(x^{'})
    &&\\\nonumber
    &\qquad+\omega^{p}~_{n\sigma}(x^{'})\mathbf{Rm}^{a}~_{bmp}(x^{'})\Big.\Big.\Big.\Big.+\omega^{c}~_{b\sigma}(x^{'})\mathbf{Rm}^{a}~_{cmn}(x^{'})-\omega^{a}~_{c\sigma}(x^{'})\mathbf{Rm}^{c}~_{bmn}(x^{'})\Big]
    &&\\\nonumber
    &\qquad+k(x,x^{'})\Big[\mathbf{Rm}^{a}~_{bnp}(x^{'})\mathbf{Rm}^{p}~_{m}(x^{'})-\mathbf{Rm}^{a}~_{bnp}(x^{'})\mathbf{Rm}^{p}_{n}(x^{'})\Big]\Big\}
    &&\\\nonumber  
    &+\frac{1}{2\pi}\int_{\sigma_{p}\,\subset\, \mathcal{G}_{p}}d\sigma_{p}\Big\{2k(x,x^{'})\xi^{\sigma}(x^{'})D_{\sigma}\mathbf{Rm}^{a}~_{bmn}(x^{'})
    \\\nonumber
    &\qquad+k(x,x^{'})\Phi(x^{'})\mathbf{Rm}^{a}~_{bmn}(x^{'})+k(x,x^{'})\xi^{\sigma}(x^{'})\Big[\mathbf{Rm}^{a}~_{bpn}(x^{'})\omega^{p}~_{m\sigma}(x^{'})\Big.\Big.
    &&\\\nonumber  
    &\qquad\Big.\Big.+\mathbf{Rm}^{a}~_{bmp}(x^{'})\omega^{p}~_{n\sigma}(x^{'})+\mathbf{Rm}^{a}~_{cmn}(x^{'})\omega^{c}~_{b\sigma}(x^{'})-\mathbf{Rm}^{c}~_{bmn}(x^{'})\omega^{a}~_{c\sigma}(x^{'})\Big]\Big\}.&&
\end{align}
Here $C_p$ is the mantle of the past light cone of $p$ and extends down to Euclidean length $\delta$ to the Cauchy hypersurface $\Sigma_{t_p-\delta}$ (this is due to the injectivity bounds). The second integration is over $\sigma_{p}\cong S^2$, which is the intersection of $C_p$ with $\Sigma_{t_p-\delta}$. $\Gamma=\Gamma(x,x')$ is the squared geodesic distance between points $x,x'\in \mathcal{G}_p$ (also known as an optical function, see Theorem 1.2.3 of Friedlander \cite{Freidlander}). $\Gamma(p,x')$ has a simple expression due to the fact that geodesics in the normal neighborhood are radial. In particular, if we let $x'^\mu$ be the normal coordinates of $x'$ and recall the $\mathcal{N}_p$ identity $\widehat{g}_{\mu\nu}(x')x'^{\mu}=\widehat{g}_{\mu\nu}(p)x'^{\mu}=\eta_{\mu\nu}x'^{\mu}$, one then concludes
\begin{align}
    \Gamma(p,x')=\eta_{\mu\nu}x'^{\mu}x'^{\nu}.
\end{align}
The measure of the $C_p$ integral is a Leray form $\mu_\Gamma$ defined by the equation $d_{x'}\Gamma(x,x')\wedge \mu_{\Gamma}(x')=\mu_{\widehat{g}}(x')$. Transforming to spherical coordinates $(t,r,\theta,\phi)$\footnote{It is important to point out that the $t$ present in the spherical coordinates is \textit{different} from the global time function defined at the beginning of \hyperlink{section.2}{Sec. 2}} allows us to find an explicit expression for $\mu_{\Gamma(p,x')}$, first notice that the metric takes a block diagonal form due to Gauss' lemma
\begin{align}
    \widehat{g}=
\begin{bmatrix}
\begin{matrix} -1 & 0 \\ 0 & 1 \end{matrix} & \begin{matrix} 0 & 0 \\ 0 & 0 \end{matrix}  \\
\:\,\begin{matrix} 0 & \:\:0 \\ 0 & \:\:0 \end{matrix} & \widehat{g}_{\theta,\phi}
\end{bmatrix}
\end{align}
Consequently one can then define null-spherical coordinates $(u=t-r,\Bar{u}=t+r,\theta,\phi)$. The geodesic squared distance between $p$ and $x'$ is then
\begin{align}
    \Gamma(p,x')=-u\Bar{u} + \widehat{g}_{\theta,\phi}(x'_{S^2},x'_{S^2})
\end{align}
which leads to $d_{x'}\Gamma(p,x')=-\Bar{u}du-ud\Bar{u}+\partial_{\theta}\Gamma d\theta+\partial_{\phi}\Gamma d\phi$, hence a solution to the defining equation for the Leray form is given by
\begin{align}
    \mu_{\Gamma(p,x')}=\frac{\sqrt{-\mathrm{det}\widehat{g}(u,\Bar{u},\theta,\phi)}}{u}du\wedge d\theta \wedge d\phi
\end{align}
Additionally, points in the mantle $C_p$ correspond to $\Bar{u}=0$ and this results in the important relation
\begin{align}
\label{eq:leray}
    \mu_{\Gamma(p,x')}|_{C_p}=\frac{1}{u=-2r}\mu_{\widehat{g}}|_{C_p}
\end{align}
The interested reader may consult Section 2.9 of \cite{Freidlander} for further discussion of Leray forms. The next relevant object in the integral equation is the symmetric transport bi-scalar $k(x,x')$ which is expressed in local coordinates as
\begin{align}
\label{eq:bi-scalar}
    k(x,x')=\frac{|\mathrm{det}(\partial^2\Gamma(x,x')/\partial x^{\mu}x'^{\nu})|^{1/2}}{4\mu^{1/2}_{\widehat{g}}(x)\mu^{1/2}_{\widehat{g}}(x')}
    \implies 
    k(p,x')=\frac{\mu^{1/2}_{\widehat{g}}(p)}{\mu^{1/2}_{\widehat{g}}(x')}.
\end{align}
The $\Phi(x')$ present in (\ref{eq:integral}) is interpreted as the dialation of $d\sigma_{p}$ along the outgoing null hypersurface emanating from $\sigma_{p}$ (denote it by $C^{+}_{\sigma_p}$), in a sense it behaves as the trace of the null second fundamental form $\Bar{\chi}$ associated to the outgoing null geodesic generator of $C^{+}_{\sigma_p}$ since it controls the evolution of the surface area of $\sigma_{p}\cong S^2$ along the outgoing null direction. Finally, $\xi$ is tangent to $C^{+}_{\sigma_p}$ and satisfies $\widehat{g}(\xi,\nabla \Gamma)=-1$ ($\xi,\nabla\Gamma$, and two orthonormal $\sigma_{p}-$tangent frame fields constitute a usual double null frame/tetrad on $\sigma_{p}$). 

\noindent The most vital property of the integral equation (\ref{eq:integral}) is that it does not contain divergence of the connection $\omega^{a}~_{b\mu}$. Using equation (\ref{eq:connection}), one may not evaluate the spacetime covariant divergence of the connections since the curvature is evaluated at a different point $\lambda x$ than the connection (evaluated at $x$).

 \noindent\textbf{Remarks about the deduction of (\ref{eq:integral}):} A solution to a wave equation on the Minkowski background can be written at each point on the spacetime in terms of the data on the intersection of the characteristics with that of a Cauchy hypersurface (past or future). This yields a representation formula for the field satisfying a wave equation, the so-called Huygens’s principle. For non-linear wave equations, instead of a representation formula, one obtains an integral equation. In the case of gravity, this is substantially more subtle since in order to write down an integral equation for any spacetime entity, one first needs to ensure that the spacetime exists in the first place. This is where we need Theorems \ref{thm:null_inj} and \ref{thm:chron_inj}. One also does not expect the validity of Huygens's principle and as a consequence, the integral equation for the curvature would contain Huygens violating tail terms that involve integration over the interior of the characteristics (light cones in relativity). This would be problematic in terms of estimating the curvature using this integral equation. This technicality is described in chapter 5 of Freidlander's book \cite{Freidlander}. Moncrief reduced the integral equations in terms of the integral over the mantle of the light cone (this is precisely the form that is provided in (\ref{eq:integral})). Concretely, one first notices that in Friedlander's analysis the choice of the Cauchy hypersurface is not fixed. Therefore, one may deform the initial topological ball (intersection of the interior of the past light cone with the initial Cauchy hypersurface) and force it inside the cone such that at the limit it can coincide with the mantle. Therefore one should expect that the integral equation can in fact be cast in terms of integrals over $C_p$ and not $D_p$. This is achieved by clever application of integration by parts together with (\ref{eq:connection}) and (\ref{eq:co-frame}).

\section{Light cone mantle estimates and proof of the main theorem}

\noindent Controlling Moncrief's integral equation (\ref{eq:integral}) forces us to examine the light cone dynamics of gravitation due to the $C_p$ term. If we want to obtain an upper bound in terms of the 0th, 1st, and 2nd order energies then we better have estimates for the mantle fluxes (turns out we only need 0th and 1st order). These will be quasi-local since $C_p\subset\mathcal{N}_p$.

\subsection{Approximate quasi-local Killing fields}

\noindent As mentioned previously, energy is not guaranteed to be a conserved quantity because we do not assume a timelike Killing field. However, there are two timelike vector fields at hand: the unit orthogonal to the CMC foliation $\mathbf{n}$ and the parallel propagated $h_{0}$. By definition, it is already known that $\widehat{g}(h_0,h_0)=\eta_{00}=-1$ so we might as well take $h_0|_p=\mathbf{n}|_p$ and, as before, parallel propagate it along the radial geodesics in the normal neighborhood of $p$\footnote{Note: $h_0|_{\mathcal{N}_p}\neq \mathbf{n}|_{\mathcal{N}_p}$ due to nontrivial holonomy effects.}. This vector field will be used to construct energy densities adapted to the light cones whereas we use only $\mathbf{n}$ for CMC slice energies.

\noindent One would expect for $h_0$ to yield approximately conserved energies within the light-cone. Moncrief \cite{integral} developed this idea in more detail and found that $h_0$ satisfies Killing's equation in the limit as one approaches $p$ radially in the normal chart, as such $h_0$ bears the name \textit{quasi-local Killing field}. This property will ultimately benefit our light cone analysis by obstructing the concentration of energy at the vertex, resulting in the desired non-blow up of the $L^\infty$ norm of the Riemann curvature.

\noindent Explicitly, from the torsion-free condition
\begin{align}
\label{eq:torsion_free_gauge_cov}
    0=D_{\mu}\Theta^{a}_{\nu}=\nabla_{\mu}\Theta^{a}_{\nu}+\omega^{a}~_{b\mu}\Theta^{b}_\nu
\end{align}
we obtain the following equation for the deformation tensor of $h_0$
\begin{align}
    ^{h_0}\pi^{\alpha\beta}&=\eta_{00}~^{\Theta^{0}}{\pi_{\mu\nu}}\widehat{g}^{\alpha\mu}\widehat{g}^{\beta\nu} \nonumber\\
    &=-\Big(\nabla_{\mu}\Theta^{0}_{\nu}+\nabla_{\nu}\Theta^{0}_{\mu}\Big)\widehat{g}^{\alpha\mu}\widehat{g}^{\beta\nu} \nonumber\\
    &=\Big(\omega^{0}~_{b\mu}\Theta^{b}_{\nu}+\omega^{0}~_{b\nu}\Theta^{b}_{\mu}\Big)\widehat{g}^{\alpha\mu}\widehat{g}^{\beta\nu}\nonumber
\end{align}
\noindent We can bound the above by means of the formulas for the connection (\ref{eq:connection}), the co-frame (\ref{eq:co-frame}), and the bootstrap assumption on the curvature (\ref{eq:bootstrap}). 
In particular, for any point $x$ in the normal neighborhood of $p$ we have
\begin{align}
\label{eq:omega_bound}
    |\omega(x)|_E\lesssim  \frac{|x|_E}{\delta^{2}}\lesssim  \frac{1}{\delta}
\end{align}
where we have used the assumption that the injectivity radius of the exponential map of $p$ is $\geq \delta$. $\Theta(p)$ is $O(1)$ by means of the orthonormality of the co-frame and the fact that $\widehat{g}|_p=\eta$ in normal coordinates
\begin{align}
    |\eta|_E=|\widehat{g}(p)|_E\sim|\eta|_E|\Theta(p)|^2_E,\quad |\eta|_E=1 \implies |\Theta(p)|_E\sim O(1) \nonumber
\end{align}
The estimates above imply that $\Theta(x)$ is also $O(1)$ throughout the normal chart
\begin{align}
\label{eq:theta_bound}
    |\Theta(x)|\leq |\Theta(p)|+\int_{0}^{1}|\omega(\lambda x)||\lambda x||\Theta(p)|d\lambda\lesssim  1+\frac{1}{\delta}\delta\sim O(1)
\end{align}
An immediate corollary is that the metric over the normal chart is also $O(1)$. Thus, we see that the deformation tensor of $h_0$ linearly goes to zero as we approach $p$ radially (recall that $x^\mu(p)=0$)
\begin{align}
    \label{eq:quasi_local_bound}
    |^{h_0}\pi(x)|\lesssim  \frac{|x|}{\delta^2} \xrightarrow{x\rightarrow p} \;0
\end{align}

\subsection{Quasi-local mantle flux estimates}

\noindent Define the 0th and 1st order $C_p$ fluxes as
\begin{align}
    \mathcal{F}^0_{C_p}:=\int_{C_p}\mathbf{Q}(h_{0},h_{0},h_{0},L)\:{\mu_{\widehat{g}}}\big|_{C_p}, \quad \mathcal{F}^{1}_{C_p}:=\int_{C_p}\mathfrak{T}^{1}(h_{0},L)\:{\mu_{\widehat{g}}}\big|_{C_p}
\end{align}
Direct calculation shows
\begin{equation}
    \begin{split}
        C^{-1}(|\alpha|^{2}+|\beta|^{2}+|\rho|^{2}+|e|^{2}+|\underline{\beta}|^{2})\leq \mathbf{Q}(h_{0},h_{0},h_{0},L)\leq C(|\alpha|^{2}+|\beta|^{2}+|\rho|^{2}+|e|^{2}+|\underline{\beta}|^{2}) \\
        C^{-1}
        \Big( |D_{L}\mathbf{Rm}_{\widehat{g}}|^{2}+\sum_{A=1,2}|D_{A}\mathbf{Rm}_{\widehat{g}}|^{2}\Big)\leq\mathfrak{T}^{1}(h_{0},L) \leq C
        \Big( |D_{L}\mathbf{Rm}_{\widehat{g}}|^{2}+\sum_{A=1,2}|D_{A}\mathbf{Rm}_{\widehat{g}}|^{2}\Big)
    \end{split}
\end{equation}
where the constants denoted by $C$ are uniform in $\delta$. An important feature is the absence of outgoing null components along $\Bar{L}$, e.g. the term $\underline{\alpha}=\mathbf{Rm}_{\widehat{g}}(\bar{L},e_A ,\bar{L},e_B )$ does not appear in the 0th order flux density, indicating that $\mathcal{F}^{0,1}_{C_p}$ only controls energy flowing across the past light cone but not along it.

\begin{proposition}
    \label{prop:general_flux_estimates}
    Let $p\in \Sigma\times[0,T^*]$ and $\mathfrak{J}$ be either $\mathbf{Q}(h_0,h_0,h_0,\cdot)$ or $ \mathfrak{T}^{1}(h_0,\cdot)$. We then obtain the following quasi-local flux estimate
    \begin{align}
        \label{eq:general_flux}
        &\mathcal{F}_{C_p}\leq E_B(t_p-\delta)+\int_{D_p}|\mathrm{div}_{\nabla}\mathfrak{J}|\:\mu_{\widehat{g}}|_{D_p}
    \end{align}
    where $E_B(t_p-\delta)$ is the energy of the $t_p-\delta$ CMC slice restricted to the ball $B_p(t_p-\delta)=D_p\cap \Sigma_{t_p-\delta}$.
\end{proposition}
\begin{proof}
    Recall that the cone with vertex $p$ extends to Euclidean length $\delta$. Stokes' theorem tells us
    \begin{align}
        \mathcal{F}_{C_p}:=\int_{C_p}\mathfrak{J}(L)\:\mu_{\widehat{g}}|_{C_p}=-\int_{B_p(t_p-\delta)}\mathfrak{J}(-\mathbf{n})\:\mu_{g}+\int_{D_p}\mathrm{div}_{\nabla}\mathfrak{J}\:\mu_{\widehat{g}}|_{D_p} \nonumber
    \end{align}
    We may repeat the same analysis of \hyperlink{section.4}{Sec. 4} applied to the quasi-local Killing field $h_0$ to say $\int_{B_p(t_p-\delta)}\mathfrak{J}(\mathbf{n})$ controls the $||D^{0,1}\mathbf{Rm}_{\widehat{g}}||^2_{L^{2}(B_p(t_p-\delta))}$ norm, which itself is controlled by $E^{0,1}(t_p-\delta)$ due to $B_p(t_p-\delta)\subset \Sigma_{t_p-\delta}$ and both being measurable sets.
\end{proof}

\begin{proposition}
    \label{prop:flux_estimates}
    Suppose $p$ is a point in the slab of Euclidean size $T^*$ and $t_p$ is its global time coordinate, then the 0th and 1st order mantle fluxes verify
    \begin{align}
        \label{eq:0th_flux}
        &\mathcal{F}^{0}_{C_p}\leq  C(^{\mathbf{n}}\pi,t_p)E^{0}(0) \\
        &\mathcal{F}^{1}_{C_p}\leq  e^{C(^{\mathbf{n}}\pi)\delta}E^{1}(t_p-\delta)+C(^{\mathbf{n}}\pi,t_p,E^{0}(0))\int_{t_p-\delta}^{t_p}||\mathbf{Rm}_{\widehat{g}}(t)||^{2}_{L^{\infty}(\Sigma_t)}\:dt
    \end{align}
    where all the $C$'s found in the RHS are independent of $\delta$.
\end{proposition}
\begin{proof}
    Apply Proposition \ref{prop:general_flux_estimates} and bound all ball integrals by CMC slice integrals. Proceed by using the 0th and 1st order energy estimates (\ref{eq:0th_energy})-(\ref{eq:1st_energy}) to arrive at the result.
\end{proof}

\subsection{Iteration argument and closure of the bootstrap}

\noindent The next step is to use Moncrief's integral equation to bound the gauge-invariant $L^\infty$ norm of the Riemann tensor in terms of the $H^2$ curvature data at an earlier time. We then substitute back into the energy estimates to perform the iteration argument described in Fig. \ref{fig:2}. Proceed by obtaining the refined point-wise bound for the curvature. In turn, this closure of bootstrap will lead to the well-posedness of the Cauchy problem.

\noindent  Before moving on, we redirect our attention to the following useful lemmas regarding the behavior of the transport bi-scalar $k(x,x')$ present in (\ref{eq:integral}).
 \begin{lemma}
 \label{lem:bi-scalar_bounds}
     The bi-scalar $k(x,x')$ verifies $\sup_{x'\in\mathcal{N}_p}|k(p,x')|\leq  1$ and hence $||k(p)||_{L^{2}(\sigma_p(t_p-\delta))}\leq \delta$.
 \end{lemma}
 \begin{proof}
     Simple. Use the local expression (\ref{eq:bi-scalar}) and the fact that the volume form is a polynomial of the components of $\widehat{g}$ which are $O(1)$ in $\mathcal{N}_p$ as seen in \hyperlink{section.5}{Sec. 5.1}.
 \end{proof}
\begin{lemma}
\label{lem:bi-scalar_derivative_bounds}
    $\sup_{x'\in D_p}|\nabla_{L,A}k(p,x')|\lesssim \delta^{-1}$.
\end{lemma}
\begin{proof}
    Do only for $\nabla_L$ as the case of $\nabla_A$ is similar. Computation in the normal coordinate system yields $x'^\beta=rNL^\beta\leq \delta NL^\beta$ (cone is of size $\delta$). We know the lapse $N$ is controlled by the deformation tensor of $\mathbf{n}$, meaning
    \begin{align}
        |\nabla_{L}k(p,x')|\leq C \delta^{-1}|x'^\beta\partial_{\beta}k(p,x')|=C\delta^{-1}|x'^\beta \frac{1}{2k(p,x')}\frac{1}{2}\widehat{g}^{\mu\nu}\partial_{\beta}\widehat{g}_{\mu\nu}|\nonumber
    \end{align}
    where (\ref{eq:bi-scalar}) and the identity $\partial_{\beta}\mu_{\widehat{g}}(x)/{\mu_{\widehat{g}}(x)}=1/2\:\widehat{g}^{\mu\nu}\partial_{\beta}\widehat{g}_{\mu\nu}$ are used to obtain the last equality. By the previous lemma, we have $|\nabla_{L}k(p,x')|\leq  C\delta^{-1}|x'^\beta \partial_\beta \widehat{g}_{\mu\nu}|$. Now invoke the following equation satisfied in normal coordinates which is actually obtained from the connection and co-frame formulas (see \cite{integral})
    \begin{align}
        x^{\beta}\partial_{\beta}\widehat{g}_{\mu\nu}&=\eta_{ab}\Bigg\{\Theta^{b}_{\nu}(x)\left(\omega^{a}_{\mu f}(x)(x^{\gamma}\Theta^{f}_{\gamma}(0))-\int_{0}^{1}\omega^{a}_{\mu f}(\lambda x)(\lambda x^{\gamma}\Theta^{f}_{\gamma}(0))\:d\lambda\right) + (a\leftrightarrow b,\mu\leftrightarrow \nu)\Bigg\}\nonumber
    \end{align}
    The normal chart estimates (\ref{eq:omega_bound}) and (\ref{eq:theta_bound}) imply the result.
\end{proof}
\begin{lemma}
\label{lem:bi-scalar_laplacian_bound}
    The bootstrap assumption for the curvature (\ref{eq:bootstrap}) implies $\sup_{x'\in D_p}|\Box k(p,x')|\lesssim  \delta^{-2}$ and hence $||\Box k(p)||^{2}_{L^2(C_p)}\lesssim  \delta^{-1}$.
\end{lemma}
\begin{proof}
    The covariant Laplacian of $k(p,x')$ is explicitly given by
    \begin{equation}\nonumber
    \begin{split}
        \Box{k(p,x')}&=\frac{1}{\mu_{\widehat{g}}(x')}\partial_{\alpha}(\mu_{\widehat{g}}(x')\widehat{g}^{\alpha\beta}\partial_{\beta}k(p,x'))
        \\
        &=-\frac{{\mu^{1/2}_{\widehat{g}}(p)}}{4\mu^{1/2}_{\widehat{g}}(x')}\widehat{g}^{\alpha\beta}\widehat{g}^{\mu\nu}\partial_{\alpha}\partial_{\beta}\widehat{g}_{\mu\nu}-\frac{{\mu^{1/2}_{\widehat{g}}(p)}}{16\mu^{1/2}_{\widehat{g}}(x')}\widehat{g}^{\alpha\beta}\widehat{g}^{\mu\nu}\widehat{g}^{ab}\partial_{\alpha}\widehat{g}_{ab}\partial_{\beta}\widehat{g}_{\mu\nu}
        \\
        &\quad\,-\frac{{\mu^{1/2}_{\widehat{g}}(p)}}{4\mu^{1/2}_{\widehat{g}}(x')}\partial_{\alpha}\widehat{g}^{\alpha\beta}\widehat{g}^{\mu\nu}\partial_{\beta}\widehat{g}_{\mu\nu}-\frac{{\mu^{1/2}_{\widehat{g}}(p)}}{4\mu^{1/2}_{\widehat{g}}(x')}\widehat{g}^{\alpha\beta}\partial_{\alpha}\widehat{g}^{\mu\nu}\partial_{\beta}\widehat{g}_{\mu\nu}.
    \end{split}
\end{equation}
The most dangerous point is the vertex i.e. $p$, where the first derivatives of the metric vanish in the normal coordinate system
\begin{equation}\nonumber
    \Box{k(p,p)}=-\frac{1}{4}\widehat{g}^{\mu\nu}\widehat{g}^{\alpha\beta}\partial_{\alpha}\partial_{\beta}\widehat{g}_{\mu\nu}.
\end{equation}
Another property of normal coordinates is that the second derivatives of the metric capture the spacetime curvature, consequently the following index symmetries hold at the origin
\begin{equation}\nonumber
    \begin{gathered}
        \partial_{c}\partial_{d}\widehat{g}_{ab}=\partial_{a}\partial_{b}\widehat{g}_{cd},
        \\
        \partial_{c}\partial_{d}\widehat{g}_{ab}+\partial_{d}\partial_{b}\widehat{g}_{ac}+\partial_{b}\partial_{c}\widehat{g}_{ad}=0.
    \end{gathered}
\end{equation}
The second derivative at the origin then satisfies
\begin{equation}\nonumber
    \partial_{\beta}\partial_{\nu}\widehat{g}_{\mu\alpha}(p)=-\frac{1}{3}(\mathbf{Rm}_{\mu\nu\alpha\beta}(p)+\mathbf{Rm}_{\alpha\nu\mu\beta}(p)),
\end{equation}
Therefore
\begin{equation}
    \Box{k(p,p)}=\frac{1}{6}\mathbf{Scal}_{\widehat{g}}(p)=0.
\end{equation}
Here $\mathbf{Scal}_{\widehat{g}}$ is the scalar curvature and it vanishes as a result of the vacuum gravity equation. The claim follows since the origin is the only possible blow up point and scaling leads to $|\Box k|\lesssim \delta^{-2}$. 
\end{proof}

\begin{proposition}
\label{prop:l_infty_bound}
    Let $p\in\Sigma\times[0,T^*]$ and $t_p$ correspond to its global time coordinate, then the gauge-invariant $L^\infty$ norm of $\mathbf{Rm}_{\widehat{g}}$ over the slice $\Sigma_{t_p}$ is bounded above by the 0th, 1st, and 2nd order energies at $t_p-\delta$. More precisely,
    \begin{align}
        ||\mathbf{Rm}_{\widehat{g}}(t_p)||^{2}_{L^{\infty}(\Sigma_{t_p})}\leq  C(^{\mathbf{n}}\pi,t_p,E^{0}(0))\Big( E^{2}(t_p-\delta)+E^{1}(t_p-\delta)+\delta^{-1}E^{0}(t_p-\delta) \Big)
    \end{align}
\end{proposition}
\begin{proof}
    The goal is to estimate (\ref{eq:integral}) evaluated at $p$ and this will suffice for the $L^\infty(\Sigma_{t_p})$ norm due to the least upper bound property. Notice that each potentially dangerous term in the $C_p$ integral belongs to one out of the following classes
    \begin{align}
        &I_{1}(p):=\int_{C_{p}}k(p,x')\omega^{d}_{e\sigma}(x')D^{\sigma}\mathbf{Rm}^{a}~_{bdf}(x')\:\mu_{\Gamma}, \nonumber
        \\
        &I_{2}(p):=\int_{C_{p}}k(p,x')\Big(\mathbf{Rm}_{fd}(x')\mathbf{Rm}_{e}~^{d}(x')-\mathbf{Rm}_{ed}(x')\mathbf{Rm}_{f}~^{d}(x')\Big)\:\mu_{\Gamma}, \nonumber
        \\
        &I_{3}(p):=\int_{C_{p}}\nabla^{\sigma}k(p,x^{'})\omega^{d}_{e\sigma}(x')\mathbf{Rm}^{a}~_{bdf}(x')\:\mu_{\Gamma},
        \\
        &I_{4}(p):=\int_{C_{p}}k(p,x^{'})\mathbf{Rm}^{a}~_{bcd}(x')\mathbf{Rm}_{ef}~^{cd}(x')\:\mu_{\Gamma}, \nonumber
        \\
        &I_{5}(p):=\int_{C_{p}}\Box{k(p,x')}\mathbf{Rm}^{a}~_{bef}(x')\:\mu_{\Gamma}. \nonumber
    \end{align}
    Furthermore, we make use of (\ref{eq:leray}) to replace the Leray forms with the canonical volume form restricted to $C_p$. Begin with $I_1(p)$. Invoke the connection formula (\ref{eq:connection}) and the fact that the normal coordinates obey $x'^\beta=rNL^\beta\leq  \delta C(^{\mathbf{n}}\pi)L^\beta$ in the cone,
    \begin{align}
        |I_1(p)|&\leq  C(^{\mathbf{n}}\pi)||k(p)||_{L^{\infty}(C_p)}\bigg| \int_{C_p}\mathcal{R}^d~_{e\sigma L}(x')D^{\sigma}\mathbf{Rm}^a~_{bdf}(x')\:\mu_{\widehat{g}}|_{C_p} \bigg| \nonumber\\
        &\leq  C(^{\mathbf{n}}\pi)\bigg| \int_{C_p}\Big(\mathcal{R}^d~_{e\Bar{L}L}(x')D_{L}\mathbf{Rm}^a~_{bdf}(x')+\sum_{A=1,2}\mathcal{R}^d~_{eAL}(x')D_{A}\mathbf{Rm}^a~_{bdf}(x')\Big)\:\mu_{\widehat{g}}|_{C_p} \bigg|\nonumber
    \end{align}
    where $\mathcal{R}^{d}~_{e\sigma L}(x'):=\int_{0}^{1}\lambda \mathbf{Rm}^{d}~_{e\sigma L}(\lambda x')\:d\lambda$ and the metric in the null basis (\ref{eq:null_metric}) has been used to obtain the last inequality as well as Lemma \ref{lem:bi-scalar_bounds}. Proceed by switching the order of integration and applying Cauchy–Schwarz,
    \begin{align}
        |I_1(p)|&\leq  C(^{\mathbf{n}}\pi)\int_{0}^{1}\lambda \Big( ||\mathbf{Rm}^d~_{e\Bar{L}L}(\lambda)||_{L^2(C_p)}||D_{L}\mathbf{Rm}^a~_{bdf}||_{L^2(C_p)} \nonumber\\
        &\qquad\qquad\qquad+\sum_{A=1,2}||\mathbf{Rm}^d~_{eAL}(\lambda)||_{L^2(C_p)}||D_{A}\mathbf{Rm}^a~_{bdf}||_{L^2(C_p)} \Big) \:d\lambda \nonumber\\
        &=C(^{\mathbf{n}}\pi)\int_{0}^{1}\lambda\cdot\lambda^{-3/2}\:d\lambda\,\Big(||\mathbf{Rm}^d~_{e\Bar{L}L}||_{L^2(C_p)}||D_{L}\mathbf{Rm}^a~_{bdf}||_{L^2(C_p)} \nonumber\\
        &\qquad\qquad\qquad\qquad\qquad\quad+\sum_{A=1,2}||\mathbf{Rm}^d~_{eAL}||_{L^2(C_p)}||D_{A}\mathbf{Rm}^a~_{bdf}||_{L^2(C_p)}\Big)\nonumber
    \end{align}
    Notice that the factor of $\lambda^{-3/2}$ is due to scaling and allows the $\lambda$ integral to be bounded. For spatial dimensions $n>3$, we would instead get $\lambda^{-n/2}$ and the $\lambda$ integral would diverge. Moreover, the components of the curvature above are precisely controlled by the 0th and 1st order fluxes. Therefore, our final estimate for $I_1(p)$ is
    \begin{align}
        |I_1(p)|^2&\leq  C(^{\mathbf{n}}\pi)\mathcal{F}^{0}_{C_p}\mathcal{F}^{1}_{C_p} \nonumber\\
        &\leq  C(^{\mathbf{n}}\pi,t_p,E^{0}(0))\Big( e^{C(^{\mathbf{n}}\pi)\delta}E^{1}(t_p-\delta)+\int_{t_p-\delta}^{t_p}||\mathbf{Rm}_{\widehat{g}}(t)||^2_{L^\infty(\Sigma_t)}dt \Big)
    \end{align}
    Moving onto $I_2(p)$, the integrand contains the antisymmetric combination $\mathbf{Rm}_{fd}\mathbf{Rm}_{e}~^{d}-\mathbf{Rm}_{ed}\mathbf{Rm}_{f}~^{d}$ which can be expanded using the metric in the null basis. The first term reads as
    \begin{equation}
        \begin{split}
            \mathbf{Rm}_{fd}\mathbf{Rm}_{e\mu}\widehat{g}^{\mu d}&=-\frac{1}{2}\Big(\mathbf{Rm}_{f\bar{L}}\mathbf{Rm}_{eL}+\mathbf{Rm}_{fL}\mathbf{Rm}_{e\bar{L}}\Big)+\sum_{A=1,2}\mathbf{Rm}_{fA}\mathbf{Rm}_{eA}
            \\
            &\sim\sum_{A=1,2}|\mathbf{Rm}_{\widehat{g}}(\cdot,\cdot,\bar{L},e_{A})|^2+(\mathrm{terms\;controlled\;by\;} \mathcal{F}^{0}_{C_p})(\mathrm{terms\;\textit{not}\;controlled\;by\;} \mathcal{F}^{0}_{C_p}) \nonumber
        \end{split}
    \end{equation}
    The expansion of $|\mathbf{Rm}_{\widehat{g}}(\cdot,\cdot,\bar{L},e_{A})|^2$ includes terms such as $|\mathbf{Rm}_{\widehat{g}}(\bar{L},e_{B},\bar{L},e_{A})|^{2}$ which do not appear in the flux density $\mathbf{Q}(h_{0},h_{0},h_{0},L)$. Fortunately, the $(f\leftrightarrow e)$ antisymmetry of the $I_2$ integrand results in a point-wise cancellation of this problematic term. This is a manifestation of the \textit{null structure}, i.e. the non-linearity responsible for the concentration of energy along the null cone (and
    a potential blow up at the vertex) turns out to be weak---Ricatti type self-interaction is absent. In the end, every time $\mathbf{Rm}_{\bar{L}A}$ appears it will be multiplied by a curvature component that is controlled by the 0th order flux $\mathcal{F}^{0}_{C_p}$  ($\mathbf{Rm}_{\bar{L}A}$ is then taken out as sup norm), therefore $I_2$ is estimated as follows
    \begin{align}
        |I_2(p)|^2&\leq ||k(p)||^2_{L^{\infty}(C_p)}\delta^{-2}\mathcal{F}^{0}_{C_p}\int_{C_p}|\mathbf{Rm}_{\widehat{g}}(x')|^2\:\mu_{\widehat{g}}|_{C_p} \nonumber\\
        &\leq  C(^{\mathbf{n}}\pi,t_p,E^{0}(0))\int_{t_p-\delta}^{t_p}||\mathbf{Rm}_{\widehat{g}}(t)||^2_{L^\infty(\Sigma_t)}dt
    \end{align}
    Apply the connection formula and Lemma \ref{lem:bi-scalar_derivative_bounds} to estimate $I_3(p)$,
    \begin{align}
        |I_3(p)|^2&\leq  C(^{\mathbf{n}}\pi)\bigg| \int_{C_p}\nabla^\sigma k(p,x')\mathcal{R}^d~_{e\sigma L}(x')\mathbf{Rm}^a~_{bdf}(x') \:\mu_{\widehat{g}}|_{C_p} \bigg|^2\nonumber \\
        &\leq  C(^{\mathbf{n}}\pi)\bigg(\int_{0}^{1}\lambda\cdot\lambda^{-3/2}\:d\lambda\bigg)^2\,\Big( ||\mathbf{Rm}^d~_{e\Bar{L}L}||^2_{L^2(C_p)}||\nabla_{L}k(p)\mathbf{Rm}^a~_{bdf}||^2_{L^2(C_p)} \nonumber\\
        &\qquad\qquad\qquad\qquad\qquad\quad +\sum_{A=1,2}||\mathbf{Rm}^d~_{eAL}||^2_{L^2(C_p)}||\nabla_{A}k(p)\mathbf{Rm}^a~_{bdf}||^2_{L^2(C_p)}\Big) \\
        &\leq  C(^{\mathbf{n}}\pi)\mathcal{F}^{0}_{C_p}||\nabla_{L,A}k(p)||^2_{L^\infty(C_p)}\delta^2\int_{t_p-\delta}^{t_p}||\mathbf{Rm}_{\widehat{g}}(t)||^2_{L^\infty(\Sigma_t)}dt \nonumber \\
        &\leq  C(^{\mathbf{n}}\pi,t_p,E^{0}(0))\int_{t_p-\delta}^{t_p}||\mathbf{Rm}_{\widehat{g}}(t)||^2_{L^\infty(\Sigma_t)}dt \nonumber
    \end{align}
    Continuing with $I_4(p)$, one can explicitly compute the quadratic curvature factor appearing on the integrand
    \begin{equation}
        \begin{split}
            \mathbf{Rm}^{a}~_{bcd}\mathbf{Rm}_{ef}~^{cd}&=\sum_{A,B=1,2}\mathbf{Rm}^{a}~_{bL {B}}\mathbf{Rm}_{ef\bar{L} {B}}+\mathbf{Rm}^{a}~_{b\bar{L}L}\mathbf{Rm}_{efL\bar{L}}+\mathbf{Rm}^{a}~_{b\bar{L} {B}}\mathbf{Rm}_{efL {B}}
            \\
            &\quad+\mathbf{Rm}^{a}~_{b {A} {B}}\mathbf{Rm}_{ef {A} {B}}+\mathbf{Rm}^{a}~_{b {A}L}\mathbf{Rm}_{ef {A}\bar{L}}+\mathbf{Rm}^{a}~_{b {A}\bar{L}}\mathbf{Rm}_{ef {A}L}+\mathbf{Rm}^{a}~_{bL\bar{L}}\mathbf{Rm}_{ef\bar{L}L} \nonumber
        \end{split}
    \end{equation}
    As in the case of $I_2$, once each term is expanded it will include at least one curvature component controlled by the light cone mantle flux and the other factor needs to be taken out as a sup norm. This means $I_{4}$ shares the same bound as $I_{2}$
    \begin{align}
        |I_4(p)|^2\leq  C(^{\mathbf{n}}\pi,t_p,E^{0}(0))\int_{t_p-\delta}^{t_p}||\mathbf{Rm}_{\widehat{g}}(t)||^2_{L^\infty(\Sigma_t)}dt
    \end{align}
    $I_5(p)$ is bounded by application of Lemma \ref{lem:bi-scalar_laplacian_bound} and Cauchy–Schwarz inequality
    \begin{align}
        |I_5(p)|^2&\leq \delta^{-2}||\Box k(p)||^2_{L^2(C_p)}\delta^{2}\int_{t_p-\delta}^{t_p}||\mathbf{Rm}_{\widehat{g}}(t)||^2_{L^\infty(\Sigma_t)}dt \nonumber\\
        &\leq  \delta^{-1}\int_{t_p-\delta}^{t_p}||\mathbf{Rm}_{\widehat{g}}(t)||^2_{L^\infty(\Sigma_t)}dt
    \end{align}
    \indent The integral over the topological 2-sphere $\sigma_p(t_p-\delta)=:\sigma$ located at the bottom of the cone is controlled by means of trace inequality. Specifically, we have
    \begin{align}
        \bigg|\frac{1}{2\pi}\int_{\sigma_p(t_p-\delta)}(\:\cdots)\;\bigg|&\leq  ||\xi||_{L^\infty(\sigma)}||k(p)||_{L^2(\sigma)}||D\mathbf{Rm}_{\widehat{g}}||_{L^2(\sigma)} \nonumber\\
        &\quad+||k(p)||_{L^\infty(\sigma)}||\Phi||_{L^2(\sigma)}||\mathbf{Rm}_{\widehat{g}}||_{L^2(\sigma)} \nonumber\\
        &\quad+||k(p)||_{L^\infty(\sigma)}||\xi||_{L^\infty(\sigma)}||\omega||_{L^2(\sigma)}||\mathbf{Rm}_{\widehat{g}}||_{L^2(\sigma)}\nonumber
    \end{align}
    Note $|\xi(x')|\sim|\Bar{L}(x')|\sim O(1)$ in the normal neighborhood and has dimension $[length]^{-1}$, whereas $\Phi(x')$ behaves like the trace of the second fundamental form associated to the outgoing null hypersurface emanating from $\sigma_p(t_p-\delta)$ thus $|\Phi(x')|\leq  C\delta^{-1}$ in the normal chart where $C$ is a uniform constant with dimension $[length]^{-1}$. Invoke (\ref{eq:omega_bound}) and Lemma \ref{lem:bi-scalar_bounds}, the estimate then reads
    \begin{align}
        \bigg|\frac{1}{2\pi}\int_{\sigma_p(t_p-\delta)}(\:\cdots)\;\bigg|&\leq  \delta ||D\mathbf{Rm}_{\widehat{g}}||_{L^2(\sigma)}+C||\mathbf{Rm}_{\widehat{g}}||_{L^2(\sigma)}+||\mathbf{Rm}_{\widehat{g}}||_{L^2(\sigma)} \nonumber\\
        &\leq  \delta^{\frac{3}{2}}||D^2\mathbf{Rm}_{\widehat{g}}||_{L^2(B(t_p-\delta))}+\delta^{\frac{1}{2}}||D\mathbf{Rm}_{\widehat{g}}||_{L^2(B(t_p-\delta))} \\
        &\quad\,+||\xi||_{L^\infty(\sigma)}\delta^{-\frac{1}{2}}||\mathbf{Rm}_{\widehat{g}}||_{L^2(B(t_p-\delta))} \nonumber
    \end{align}
    The last inequality follows from Stokes' theorem and scaling (one must check that both sides of the inequality share the same dimension, $||\xi||_{L^{\infty}}\lesssim 1$). The presence of $\delta^{-\frac{1}{2}}$ could potentially be problematic. However, we will see this dangerous term appears multiplied by $\delta$ in the later estimates.

    \noindent Combining all the bounds yields
    \begin{align}
        ||\mathbf{Rm}_{\widehat{g}}(t_p)||^2_{L^\infty(\Sigma_{t_p})}&\leq  \delta^3E^{2}(t_p-\delta)+\delta E^{1}(t_p-\delta)+\delta^{-1} E^{0}(t_p-\delta) \nonumber \\
        &\quad+C(^{\mathbf{n}}\pi,t_p,E^{0}(0))\Big( e^{C(^{\mathbf{n}}\pi)\delta}E^{1}(t_p-\delta)+(1+\delta^{-1})\int_{t_p-\delta}^{t_p}||\mathbf{Rm}_{\widehat{g}}(t)||^2_{L^\infty(\Sigma_t)}dt \Big) \nonumber
    \end{align}
    Apply Grönwall's inequality for the last time and safely replace all positive exponents of $\delta$ with $\leq  1$
    \begin{align}
        ||\mathbf{Rm}_{\widehat{g}}(t_p)||^2_{L^\infty(\Sigma_{t_p})}&\leq  C(^{\mathbf{n}}\pi,t_p,E^{0}(0))\Big( E^{2}(t_p-\delta)+E^{1}(t_p-\delta)+\delta^{-1}E^{0}(t_p-\delta) \Big)e^{C'(^{\mathbf{n}}\pi,t_p,E^{0}(0))(1+\delta^{-1})\delta} \nonumber \\
        &\leq  C(^{\mathbf{n}}\pi,t_p,E^{0}(0))\Big( E^{2}(t_p-\delta)+E^{1}(t_p-\delta)+\delta^{-1}E^{0}(t_p-\delta) \Big)
    \end{align}
    which is the result we wanted to show.
\end{proof}
\noindent We are now in place to apply the iteration argument to bound the CMC energies at $t^*$ in terms of the initial $H^2$ data. Again, the idea is to run the global energy estimates and quasi-local light cone estimates until we reach $t=0$.

\begin{proposition}
\label{prop:h^2_bound}
    Suppose $t^*$ is a time close to $T^*$, then
    \begin{align}
        E^{2}(t^*)+E^{1}(t^*)\leq  1+\frac{C(^{\mathbf{n}}\pi,t^*,||\mathbf{Rm}_{\widehat{g}}||_{H^{2}(\Sigma_{t=0})})}{\delta}
    \end{align}
    here $C$ does not depend on $\delta$ and depends only on the bound for $^{\mathbf{n}}\pi$ (\ref{eq:breakdown_criteria}), the time $t^*$, and the initial hypersurface $H^2$ data of the curvature.
\end{proposition}
\begin{proof}
    Combine the results from Propositions \ref{prop:energy_estimates} and \ref{prop:l_infty_bound} to begin bounding the energies at $t^*$ in terms of the data at $t^*-\delta$ up to the $L^\infty$ of Riemann which is taken out of the $[t^*-\delta,t^*]$ integral as a sup norm. We gain a factor of $\delta$ which annihilates the $\delta^{-1}$ in front of the 0th order energy
    \begin{align}
        E^{2}(t^*)+E^{1}(t^*)&\leq  E^{2}(t^*-\delta)+E^{1}(t^*-\delta)+\sup_{t\in[t^*-\delta,t^*]}\Big(\delta E^{2}(t-\delta)+\delta E^{1}(t-\delta)+E^{0}(t-\delta) \Big) \nonumber\\
        &\quad+\sup_{t\in[t^*-\delta,t^*]}\Big(\delta E^{2}(t-\delta)+\delta E^{1}(t-\delta)+E^{0}(t-\delta) \Big)^2+\Big(E^{1}(t^*-\delta)\Big)^2 \nonumber
    \end{align}
    Since the interval $[0,t]$ is finite for any $t\leq T^*$, we may always cover it with a finite number of copies of $\delta$ intervals, say $\mathbf{K}(t)$ of them. In other words, $\exists \mathbf{K}(t)$ such that $t-\mathbf{K}(t)\delta=0$. Furthermore, the 0th order energy is uniformly bounded in $\delta$ in terms of $E^{0}(0)$, $C(^{\mathbf{n}}\pi)$, and $t^*$
    \begin{align}
        E^{2}(t^*)+E^{1}(t^*)&\leq  E^{2}(t^*-\delta)+E^{1}(t^*-\delta)+\sup_{t\in[t^*-\delta,t^*]}\Big(\delta E^{2}(t-\delta)+\delta E^{1}(t-\delta) \Big) \nonumber\\
        &\qquad+\sup_{t\in[t^*-\delta,t^*]}\Big(\delta E^{2}(t-\delta)+\delta E^{1}(t-\delta) \Big)^2 + C\nonumber \\
        &\leq  e^{C(^{\mathbf{n}}\pi)\delta}E^{2}(t^*-2\delta)+e^{C(^{\mathbf{n}}\pi)\delta}E^{1}(t^*-2\delta)+\sup_{t\in[t^*-2\delta,t^*-\delta]}\Big(\delta E^{2}(t-\delta)+\delta E^{1}(t-\delta) \Big) \nonumber\\
        &\qquad+\sup_{t\in[t^*-2\delta,t^*-\delta]}\Big(\delta E^{2}(t-\delta)+\delta E^{1}(t-\delta) \Big)^2 +\sup_{t\in[t^*-\delta,t^*]}\Big(\delta E^{2}(t-\delta)+\delta E^{1}(t-\delta) \Big) \nonumber\\
        &\qquad+\sup_{t\in[t^*-\delta,t^*]}\Big(\delta E^{2}(t-\delta)+\delta E^{1}(t-\delta) \Big)^2 + 2C\nonumber \\
        &\leq e^{C(^{\mathbf{n}}\pi)\mathbf{K}(t^{*})\delta} E^{2}(0)+e^{C(^{\mathbf{n}}\pi)\mathbf{K}(t^{*})\delta}E^{1}(0)+\mathbf{K}(t^*)C+\sup_{t\in[t^*-\delta,t^*]}\Big(\delta E^{2}(t-\delta)+\delta E^{1}(t-\delta)\Big) \nonumber\\
        &\qquad\qquad+\;\cdots\;+\sup_{t\in[t^*-\mathbf{K}(t^*)\delta,t^*-(\mathbf{K}(t^*)-1)\delta]}\Big(\delta E^{2}(t-\delta)+\delta E^{1}(t-\delta)\Big) \nonumber \\
        &\qquad\qquad+\sup_{t\in[t^*-\delta,t^*]}\Big(\delta E^{2}(t-\delta)+\delta E^{1}(t-\delta)\Big)^2 \nonumber\\
        &\qquad\qquad+\;\cdots\;+\sup_{t\in[t^*-\mathbf{K}(t^*)\delta,t^*-(\mathbf{K}(t^*)-1)\delta]}\Big(\delta E^{2}(t-\delta)+\delta E^{1}(t-\delta)\Big)^2 \nonumber
    \end{align}
    Re-run estimates on the energies inside the supremums and note that it takes up to $\mathbf{K}(t^*)$ steps to reach initial energies. The $\sup$ terms contain positive powers of $\delta$. In the end, we obtain the following
    \begin{align}
        E^{2}(t^*)+E^{1}(t^*)&\leq  \Big( E^{2}(0)+E^{1}(0)\Big)e^{C'\mathbf{K}(t^*)\delta}++\mathbf{K}(t^*)C+(C''\delta)^{\mathbf{K}(t^*)} \nonumber \\
        &\leq  Ce^{C\mathbf{K}(t^*)\delta}+\mathbf{K}(t^*)C+(C\delta)^{\mathbf{K}(t^*)}\leq  1+\frac{C(^{\mathbf{n}}\pi,t^*,||\mathbf{Rm}_{\widehat{g}}||_{H^{2}(\Sigma_{t=0})})}{\delta}
    \end{align}
    where all the constants involved depend on the deformation tensor bound, the time $t^*$, and the initial slice curvature data.
\end{proof}
\noindent Now we complete the bootstrap argument. 
\begin{theorem}
\label{thm:bootstrap_closure}
    The gauge-invariant $L^\infty$ norm of $\mathbf{Rm}_{\widehat{g}}$ over the slab $\Sigma\times[0,T^*]$ has a refined point-wise upper bound dependent only on the initial data. Specifically,
    \begin{align}
        \sup_{t\in[0,T^*]}\lvert\lvert{\mathbf{Rm}_{\widehat{g}}(t)}\rvert\rvert^{2}_{L^{\infty}(\Sigma_t)}\leq  C(^{\mathbf{n}}\pi,T^*,||\mathbf{Rm}_{\widehat{g}}||_{H^{2}(\Sigma_{t=0})})
    \end{align}
    where the RHS does not depend on $\delta$ at all.
\end{theorem}
\begin{proof}
    The last two propositions, namely \ref{prop:l_infty_bound} and \ref{prop:h^2_bound}, together with the uniform bound for the 0th order energy yields
    \begin{align}
        \sup_{t\in[0,T^*]}||\mathbf{Rm}_{\widehat{g}}(t)||^2_{L^\infty(\Sigma_t)}\leq  C_1+C_2\delta^{-1}
    \end{align}
    where $C_{1}$ and $C_{2}$ do not depend on the $\delta$. $C_{1}+C_{2}\frac{\delta}{\delta^{2}}$ can be made to be smaller than $\frac{1}{4\delta^{4}}$ (the boot-strap assumption made in \ref{eq:bootstrap}) after choosing a sufficiently small but fixed $\delta>0$ that only depends on the initial $H^2$ curvature data, the time $T^*$, and $C(^{\mathbf{n}}\pi)$ (i.e., $\delta^{2}\leq \frac{1-4C_{2}\delta}{4C_{1}}$). Therefore $\sup_{t\in[0,T^*]}\lvert\lvert{\mathbf{Rm}_{\widehat{g}}(t)}\rvert\rvert^{2}_{L^{\infty}(\Sigma_t)}\leq  C(^{\mathbf{n}}\pi,T^*,||\mathbf{Rm}_{\widehat{g}}||_{H^{2}(\Sigma_{t=0})})$.
\end{proof}
\begin{theorem}
    The $H^2(\Sigma_t)$ norm of the Riemann curvature tensor at any $t\in [0,T^*]$ is bounded above by a finite quantity dependent \textit{only} on the slab length $T^*$, the initial $H^2$ curvature data, and the assumed point-wise bound on the deformation tensor of the unit timelike vector field orthogonal to the CMC Cauchy foliation
    \begin{align}
        ||\mathbf{Rm}_{\widehat{g}}(t)||_{H^{2}(\Sigma_t)}\leq C(^{\mathbf{n}}\pi,T^*,||\mathbf{Rm}_{\widehat{g}}||_{H^{2}(\Sigma_{t=0})}) <\infty
    \end{align}
\end{theorem}
\begin{proof}
    Using Theorem \ref{thm:bootstrap_closure}, we apply the energy estimates of Proposition \ref{prop:energy_estimates} over the time interval $[0,t]$, this yields
    \begin{align}
        E^{1}(t)+E^{0}(t)&\leq e^{C(^{\mathbf{n}}\pi)t}\Big(E^{1}(0)+E^{0}(0)+C(^{\mathbf{n}}\pi,T^*,t,||\mathbf{Rm}_{\widehat{g}}||_{H^{2}(\Sigma_{t=0})})t\Big) \nonumber\\
        &\leq C(^{\mathbf{n}}\pi,T^*,||\mathbf{Rm}_{\widehat{g}}||_{H^{2}(\Sigma_{t=0})})\nonumber
    \end{align}
    where we have used the assumption that $t\leq T^*$ and the monotonicity of the exponential and linear functions in order to get rid of the $t$ dependence on the RHS. The estimate for the 2nd order energy at $t$ is
    \begin{align}
        E^{2}(t)&\leq e^{C(^{\mathbf{n}}\pi)t}\Big(E^{2}(0)+C(^{\mathbf{n}}\pi, T^*, ||\mathbf{Rm}_{\widehat{g}}||_{H^{2}(\Sigma_{t=0})})\sup_{t'\in[0,t]}E^{1}(t') \,t\Big) \nonumber \\
        &\leq C(^{\mathbf{n}}\pi,T^*,||\mathbf{Rm}_{\widehat{g}}||_{H^{2}(\Sigma_{t=0})})\nonumber
    \end{align}
    The sum $E^{2}(t)+E^{1}(t)+E^{0}(t)$ controls $||\mathbf{Rm}_{\widehat{g}}(t)||_{H^{2}(\Sigma_t)}$, thus implying the result.
\end{proof}

\noindent Once we have established that the $H^{2}(\Sigma)$ norm of the Riemann curvature cannot blow up in any finite time interval $[0,T^*]$, the remaining task is to use the necessary elliptic estimates. Using curvature bounds, one obtains necessary bound for the volume $\mathrm{Vol}_{\widehat{g}}(\mathcal{B}_{E}(p,C\delta))$. In terms of the elliptic estimates, we refer to \cite{andersson2004future, klainerman2010breakdown}. An application of the local well-posedness theorem of \cite{elliptichyperbolic} in CMC spatial harmonic coordinates for spacetimes foliated by compact Cauchy slices of negative Yamabe type together with the arguments presented in \cite{lefloch2008injectivity} regarding recovery of the spacetime given curvature bound yields the main theorem. 

\noindent \textbf{Remark:} \textit{At this point, we want to conjecture a continuation criteria for Einstein equations coupled with sources. This is based on the preliminary analysis of the elliptic equation for the lapse function in the CMC gauge. To satisfy one of the criteria in LeFloch and Chen's theorem, one has to obtain a point-wise bound on the lapse function. Let $\mathfrak{TT}:=\mathfrak{TT}_{\mu\nu}dx^{\mu}\otimes dx^{\nu}$ be the stress energy tensor of a source in the Einsteinian spacetime. The lapse function $N$ verifies the following elliptic equation in the CMC gauge
\begin{eqnarray}
\Delta_{g}N+\{|k|^{2}+\frac{S}{n-1}+\frac{n-2}{n-1}E\}N=\frac{\partial \mathrm{tr}_{g}k}{\partial t},
\end{eqnarray}
where $E:=\mathfrak{T}\mathfrak{T}(\mathbf{n},\mathbf{n})$ is the energy density and $S:=g^{ij}(\mathfrak{TT}(\partial_{i},\partial_{j}))$ is the trace of the momentum flux density. Maximum principle yields
$\frac{|\frac{\partial \mathrm{tr}_{g}k}{\partial t}|}{||k||^{2}_{L^{\infty}}+||S||_{L^{\infty}}+||E||_{L^{\infty}}}\leq  ||N||_{L^{\infty}}\leq  |\frac{\partial \mathrm{tr}_{g}k}{\partial t}|$. Therefore, the preliminary continuation criterion would be the finiteness of $||~^{\mathbf{n}}\pi||_{L^{\infty}}+||S||_{L^{\infty}}+||E||_{L^{\infty}}$.  For example,\\
(a) in coupled Einstein-Yang-Mills system one would expect that a preliminary continuation criteria would be the boundedness of $||~^{\mathbf{n}}\pi||_{L^{\infty}}+||F||_{L^{\infty}}$ ($F$ is the Yang-Mills curvature).\\
(b) in coupled Einstein-Euler system, the preliminary continuation criteria can be cast as boundedness of $||~^{\mathbf{n}}\pi||_{L^{\infty}}+||(P,\rho,v)||_{L^{\infty}}$, where $P,\rho,$ and $v$ are the pressure, density, and velocity of the fluid.\\
Of course, one ought study the coupled light cone dynamics  (sound cone as well for fluid) which is essentially the second important step.}

\section{Discussion}

\noindent  The weak cosmic censorship conjecture asserts that all spacetime singularities are hidden inside black holes, i.e. the future null infinity is geodesically complete. In a globally hyperbolic vacuum with non-compact Cauchy hypersurfaces, Penrose's singularity theorem \cite{hawking1970singularities} states that spacetime cannot be future null complete whenever there exists a closed trapped surface. Thus if the weak cosmic censorship hypothesis holds, then it would be impossible for naked singularities to occur. Addressing the question of cosmic censorship is important due to the potential pathologies mentioned in \hyperlink{section.1}{Sec. 1}, but it also plays a vital role in other contexts where it is assumed to be true e.g. black hole radiation. The first step to finding a solution to the conjecture is to determine the breakdown criteria for solutions to the Einstein equations. Although global existence is not known to hold for general gravitational fields due to instability issues that can lead to the formation of black holes, a vacuum is the easiest setting for determining the exact conditions that cause a breakdown of solutions. We tackled this problem using physical principles as backbones to motivate what the continuation criteria should be, namely the $L^{\infty}_{t}L^{\infty}_{\vec{x}}$ bound for the deformation tensor of the unit timelike vector field normal to the CMC Cauchy hypersurfaces. Inspiration was taken from the
general Minkowski space Yang-Mills global existence work by Eardley and Moncrief \cite{eardley1982global2}, but with
the caveat that the spacetime curvature takes the role of the field strength, meaning the background
geometry is not fixed whereas it would be in non-gravitational field theories. The frame bundle formalism of GR led to formulae (connection, co-frame field) that served a crucial role in arriving at the necessary estimates. Of course, injectivity bounds were needed to consider these equations in the first place.        

One would however wonder the physical appeal of our result. Recall that the deformation tensor $^{\mathbf{n}}\pi:=\mathcal{L}_{\mathbf{n}}\widehat{g}$ measures the obstruction of $\mathbf{n}$ to be a timelike Killing vector field. This only requires certain derivatives of the spacetime metric. Therefore a point-wise bound on $^{\mathbf{n}}\pi$ is rather quite rough in the sense that $||^\mathbf{n}\pi||_{L^{\infty}_{t,x}}<\Delta<\infty$ is a larger space that includes the space of classical solutions as a smaller subspace. However, in a quantum theory of gravity, one ought to integrate over the space of Riemannian metrics, not just the smooth ones solving the Einstein's equations \cite{gibbons1993euclidean}. The question then becomes which function space of metrics ought to be considered in the definition of the partition function? It is matter of debate since these infinite dimensional moduli spaces are in general difficult to handle in mathematically rigorous way. A reasonable choice would be the space where classical determinism is valid. In general one would not expect to have a smooth geometry at quantum level and therefore a rough spacetime metric is often desirable (there are several propositions of lattice structure of quantum spacetimes \cite{loll2019quantum, grimmer2022discrete}).

\noindent In addition, the bound on the deformation tensor is a key obstruction to proving the weak cosmic censorship hypothesis. Even though we only considered spacetimes that are foliated by the compact Cauchy hypersurfaces, one may obtain the same result for asymptotically flat spacetimes foliated by maximal slices (the argument behind proving the point-wise bound of the spacetime curvature mostly remains unchanged due to its quasi-local nature). However, the deformation tensor bound seems to be reasonable as most realistic physical systems exhibit an almost-timelike symmetry at least in domains of outer communication. Of course, one can never have any a priori information about this deformation tensor on a dynamical spacetime and as such it may blow up in finite time, hence obstructing the continuation of the solution. Leaving aside the \textit{generic} spacetimes, one may wonder if the global existence result can be proven to be true in certain special spacetimes for the large data. Positive answers to this question are available for Gowdy spacetimes where a $\mathbb{T}^{3}$ symmetry is present \cite{chrusciel1990strong, moncrief1981global}. A quantum jump in the context of large data global existence result would be that of $\mathrm{U}(1)$ problem where a small data result is already established \cite{choquet2003nonlinear, choquet2001future}. Another important direction would be to study large data Einstein-Yang-Mills dynamics without symmetry assumption since Yang-Mills repulsion is expected to counterbalance gravity.

\end{document}